\documentclass[a4paper,UKenglish,cleveref,autoref,thm-restate]{lipics-v2021}

\usepackage{algorithm}
\usepackage{algpseudocode}
\usepackage[table]{xcolor}
\usepackage{tikz}
\usepackage{pgfplots}
\usepackage{subcaption,ragged2e}
\usepackage{afterpage}
\usepackage{adjustbox}
\pgfplotsset{compat=1.18}

\newcommand{\ourname}{Proof-of-Theft}
\newcommand{\ournameshort}{PoT}
\newcommand{\ourfullname}{{\ournameshort} (\textit{{\ourname}})}

\pdfoutput=1
\hideLIPIcs

\title{Dynamic Graph-based Fingerprinting of In-browser Cryptomining}

\titlerunning{Dynamic Graph-based Fingerprinting}

\author{Tanapoom {Sermchaiwong}}{The Hong Kong University of Science and Technology, Hong Kong, China}{tanapoom.se@gmail.com}{}{}

\author{Jiasi Shen}{The Hong Kong University of Science and Technology, Hong Kong, China}{sjs@cse.ust.hk}{}{}

\authorrunning{T. Sermchaiwong et al.}

\Copyright{Tanapoom Sermchaiwong and Jiasi Shen}



\category{} 

\relatedversion{}

\acknowledgements{We would like to thank Sizhe Zhong and Jingyi Chen for their valuable feedback.}

\nolinenumbers 

\EventEditors{John Q. Open and Joan R. Access}
\EventNoEds{2}
\EventLongTitle{42nd Conference on Very Important Topics (CVIT 2016)}
\EventShortTitle{CVIT 2016}
\EventAcronym{CVIT}
\EventYear{2016}
\EventDate{December 24--27, 2016}
\EventLocation{Little Whinging, United Kingdom}
\EventLogo{}
\SeriesVolume{42}
\ArticleNo{23}

\begin{document}

\maketitle

\begin{abstract}
The decentralized and unregulated nature of cryptocurrencies, combined with their monetary value, has made them a vehicle for various illicit activities. One such activity is cryptojacking, an attack that uses stolen computing resources to mine cryptocurrencies without consent for profit. In-browser cryptojacking malware exploits high-performance web technologies like WebAssembly to mine cryptocurrencies directly within the browser without file downloads. Although existing methods for cryptomining detection report high accuracy and low overhead, they are often susceptible to various forms of obfuscation, and due to the limited variety of cryptomining scripts in the wild, standard code obfuscation methods present a natural and appealing solution to avoid detection. To address these limitations, we propose using instruction-level data-flow graphs to detect cryptomining behavior. Data-flow graphs offer detailed structural insights into a program's computations, making them suitable for characterizing proof-of-work algorithms, but they can be difficult to analyze due to their large size and susceptibility to noise and fragmentation under obfuscation. We present two techniques to simplify and compare data-flow graphs: (1) a graph simplification algorithm to reduce the computational burden of processing large and granular data-flow graphs while preserving local substructures; and (2) a subgraph similarity measure, the \textit{n-fragment inclusion score}, based on fragment inclusion that is robust against noise and obfuscation. Using data-flow graphs as computation fingerprints, our detection framework {\ourfullname} was able to achieve high detection accuracy against standard obfuscations, outperforming existing detection methods. Moreover, {\ournameshort} uses generic data-flow properties that can be applied to other platforms more susceptible to cryptojacking such as servers and data centers.
\end{abstract}

\section{Introduction}

A cryptocurrency is a decentralized peer-to-peer digital exchange system that functions as a currency without a central authority \cite{7906988}. The absence of a central governing body and the promise of freedom and resistance to censorship have led to widespread initial adoption of cryptocurrencies. Combined with the speculative nature of the market, this has created an enormous growth in demand \cite{chaincatcherDataAdoption}. The staggering increase in the monetary value of many cryptocurrencies in recent years has made cryptomining a potential source of revenue derived from computing power \cite{farell2015analysis}. In some cases, stolen computing resources are used to generate profits through unauthorized mining. The use of unauthorized computational power to mine cryptocurrency without consent is called \textit{cryptojacking}, and this attack happens on all scales of computing systems, from large data centers to small personal platforms such as web browsers and IOT devices \cite{arunkumar2024review}, \cite{tekiner2022lightweight}. As recently as 2024, instances of cryptojacking attacks have caused significant monetary damages to organizations such as the United States government \cite{fedscoopEvenGovernment} and many financial institutions \cite{coindeskCryptojackingFinancial}.

In-browser cryptojacking has been described as a recently emerged type of fileless malware that is difficult to detect in the traditional framework of malware detection \cite{carlin2019you}. Recent studies have demonstrated the prevalence of in-browser cryptojacking on popular websites based on their analyses of top-ranking websites \cite{eskandari2018first}, \cite{musch2019thieves}, \cite{tekiner2021browser}, \cite{9092245}. In-browser cryptojacking have been enabled partially by new technologies such as WebAssembly (Wasm) \cite{webassemblyWebAssembly} and WebWorkers \cite{mozillaUsingWorkers}, which were introduced to facilitate high-performance applications to run on web browsers, delivering cryptominers filelessly through scripts on a webpage. In fact, previous works showed that most of these miners are implemented in Wasm \cite{10.1145/3243734.3243858}. Popular services such as CoinHive facilitated this process by providing mining scripts and mining pools as an alternative method of revenue generation for websites, and although the CoinHive service was shut down in 2019, multiple studies indicate that in-browser cryptojacking is still prevalent in the wild \cite{eskandari2018first}, \cite{musch2019thieves}, \cite{tekiner2021browser}.

To prevent theft of computing resources, various techniques have been proposed to detect and prevent cryptominers from running in browsers. These include traditional methods such as domain name blocking and keyword blacklists \cite{githubGitHubHoshsadiqadblocknocoinlist}, \cite{githubGitHubKerafNoCoin}, \cite{githubGitHubGorhilluBlock}, as well as those involving more advanced analysis such as semantic instruction counting \cite{wang2018seismic}, \cite{bian2020minethrottle}, \cite{8514167}, CPU, memory, and network traffic monitoring \cite{rodriguez2018rapid}, \cite{kelton2020browser}, \cite{kharraz2019outguard}, and static approaches based on machine learning and deep learning \cite{naseem2021minos}, \cite{romano2020wasim}. Although existing detection methods report high accuracy, they are often susceptible to obfuscation \cite{10.1145/3507657.3528560}, \cite{harnes2024cryptic}, \cite{wang2018seismic}. For example, proxies, dynamically generated domain names, and encrypted WebSocket communication can render blacklists and network-based detection methods ineffective. Statistical distributions of the instruction count can be skewed by performance throttling and the insertion of spurious instructions. CPU and memory event monitoring are susceptible to noise from other processes or web pages \cite{9306696}. Deep learning-based approaches such as those proposed by MINOS \cite{naseem2021minos} and WASim \cite{romano2020wasim} have been shown to perform poorly on obfuscated and diversified binaries \cite{harnes2024cryptic}, \cite{10.1145/3507657.3528560}. These obfuscations are easy to apply with access to the source code, thereby restricting the usefulness of prior approaches due to how straightforward they are to bypass. There remains a large gap in detecting obfuscated miners effectively.

Cryptomining in a single browser is often too slow to generate profit since the probability of calculating the correct hash in a reasonable amount of time is minuscule, and the solution is usually to mine cryptocurrency as a part of a larger pool, where profit from any correctly mined hash is shared among the participants. This suggests that in-browser cryptojacking is only possible on economies of scale, and numerous studies support this hypothesis, indicating that a large majority of in-browser cryptomining scripts originate from a limited number of services (e.g., CoinHive, CoinImp, JSECoin) which provide the necessary infrastructure such as the mining scripts and pool \cite{tekiner2021browser}, \cite{9092245}, \cite{kharraz2019outguard}. The small diversity of cryptominers deployed on a large number of platforms makes obfuscation an attractive solution to evade detection. Therefore, there is a need for better detection methods of obfuscated cryptomining malware.

In this paper, we propose {\ourfullname}, a new approach to detect obfuscated cryptomining malware based on the following key insights. Fundamentally, a cryptominer performs calculations to validate transactions on the blockchain. In most proof-of-work schemes, this entails repeated hashing of a block of data to generate a hash satisfying an arbitrary but difficult property \cite{7906988}. This repetitive computation is an intrinsic property of a proof-of-work scheme. To characterize the computations performed by an algorithm, instruction-level data-flow graphs provide a structured and comprehensive view of computation that is difficult to manipulate. We hypothesize that the data-flow graphs provide us with a direct view of a cryptominer's core characteristic. Moreover, code obfuscators are known to operate within certain boundaries \cite{272274}, suggesting that there are a limited number of transformations they can make to the data-flow properties of a program.

Instruction-level data-flow graphs are difficult to compare due to their sizes and susceptibility to noise and fragmentation. To enable their use in cryptomining detection, we propose a set of graph analysis techniques consisting of: (1) a graph simplification algorithm to generate computation fingerprints from data-flow graphs; and (2) a subgraph similarity measure to search for malicious behavior in fingerprints. We demonstrate that {\ournameshort} outperforms the state-of-the-art in cryptominer detection under various obfuscations. To the best of our knowledge, this study is the first to utilize instruction-level data-flow graphs in either detecting cryptominers or software classification in general, and while we implement our algorithms and perform the experiments on the WebAssembly platform for web browsers, our theoretical framework for data-flow graph analysis and computation detection uses generic data-flow properties that can be applied to other platforms more susceptible to cryptojacking such as servers and data centers.

In summary, this paper makes the following three major contributions:
\begin{itemize}
    \item We present a novel algorithm to simplify large repetitive data-flow graphs that preserves local substructures, enabling large and granular data-flow graphs to be used in cryptomining detection.
    \item We introduce a new subgraph similarity measure, \textit{n-fragment inclusion score}, to compare graphs based on inclusion that is resilient against noise, fragmentation, and obfuscation.
    \item We implement and evaluate {\ournameshort} for detecting WebAssembly cryptominers on a sample of 29 real-world web applications, 6 cryptominers, and 30 obfuscated cryptominers, showing the effectiveness of instruction-level data-flow graphs in obfuscation resistant detection of cryptojacking.
\end{itemize}

\section{Background and Motivation}

In this section, we present key concepts that are essential for the remainder of the paper, as well as the motivation behind our methodology. First, we provide an overview of proof-of-work cryptocurrencies and how they motivate our analysis. We give a brief review of the current literature on cryptominer detection and the role of resource graphs in malware detection to demonstrate the need for a better detection framework. Finally, we explain key details of the WebAssembly binary format and the data-flow graphs we collect using dynamic analysis.

\subsection{Cryptomining and Proof-of-Work Schemes}

To ensure the validity of transactions and prevent malicious agents from compromising the integrity of the blockchain, a consensus mechanism is employed by cryptocurrency systems to validate the authenticity of new transactions \cite{7906988}. This consensus mechanism involves multiple users validating the ownership and transfer of currency, ensuring that one may only spend the currency in their possession. The ownership information necessary for validation exists transparently on the public ledger in the blockchain. To prevent a single actor from creating a false consensus by using multiple nodes to validate an invalid transaction, a validator (miner) is required to provide proof that they possess a certain amount of resource which acts as a form of artificial cost or barrier to entry. The three most ubiquitous proofs are as follows.
\begin{itemize}
    \item \textit{Proof of Work.} The miners are asked to perform a resource intensive computation that is easy to verify. This proof presents a barrier of entry by requiring each validation to be backed by a certain amount of computing resources.
    \item \textit{Proof of Stakes.} The validators are chosen in proportion to the amount of currency they hold.
    \item \textit{Proof of Retrievability.} The validators must prove that they can store a large piece of data intact and able to retrieve it at will. The proof creates a barrier of entry by requiring a large storage capacity.
\end{itemize}
To incentivize users and stakeholders to participate in securing the blockchain, miners are given a reward for each block they successfully validate. In proof-of-work schemes, this is commonly referred to as cryptomining, and the compensation for successful block validations generates revenue for the miner.

A recent study shows that proof-of-work cryptocurrencies dominate over 57\% of the market share \cite{BAJRA2024102571}. Our study focuses on such schemes, where computing resources can be used to generate revenue, as is central to cryptojacking. The computationally intensive task required to verify a block usually involves calculating a cryptographic hash function on the transaction data combined with a randomized value until a hash with an arbitrary but difficult property is found. For example, the Bitcoin \cite{nakamoto2008bitcoin} currency requires miners to compute a \verb|sha2| hash that is numerically smaller than the network's difficulty target \cite{antonopoulos2014mastering}.

Several proof-of-work algorithms have been proposed to solve different challenges faced by decentralized currencies. A major issue faced by early cryptocurrencies, such as Bitcoin, is that the hash function used to validate blocks favors performance on application-specific integrated circuits (ASICs) and GPUs, leading to the consolidation of mining power in large-scale farms and discouraging average users from mining due to the inefficiency of using a consumer CPU, possibly threatening the decentralized nature of the blockchain network \cite{8516911}. Later cryptocurrencies adopt hashing functions, for instance yescrypt \cite{openwallYescryptScalable} and CryptoNight \cite{getmoneroCryptoNightMonero}, with ASIC and GPU resistant properties such as memory-hardness, where a memory bottleneck diminishes the computational efficiency of ASICs and GPUs. Cryptocurrencies designed for consumer CPUs such as Monero \cite{getmoneroMoneroProject}, MintMe \cite{mintmeMintMeCreate}, and other CryptoNight coins have been the common currencies used for in-browser cryptojacking \cite{tekiner2021browser}, \cite{coinimpCoinIMPJavaScript}, as they are the most efficient currencies to mine in Wasm.

A fundamental characteristic of most proof of work algorithms used in cryptocurrency systems is an extensive amount of repetitive and artificial computation. For example, the Bitcoin mining algorithm is a simple search problem where the random nonce represents the search space, and each search operation requires a nontrivial computation of the \verb|sha2| hash function. In a data-flow graph where each execution of an instruction represents a unique vertex, redundant computations emerge as repeated subgraphs representing the same computation performed many times. This key insight allows us to compress the graph into a much smaller form while preserving local semantics of a program.

\subsection{Cryptojacking Detection}

Numerous studies have proposed highly accurate detection systems to address cryptojacking malware. These systems utilize one or more of the following program features: network behavior, resource and performance metrics, semantic instruction count, and the program binary file. We discuss these systems in the following section.

\noindent\\
\textbf{Network Behavioral Detection.} Several network detection methods have been proposed for large-scale and platform-independent detection of cryptojacking. Caprolu et al.~\cite{CAPROLU2021126} and Pastor et al. \cite{9178288} proposed the use of network flow features (i.e. packet sizes and inter-arrival time) to classify and detect cryptominers which communicate with a mining pool. MineCap \cite{minecap} presented a similar idea using super-incremental learning to reduce the training burden. XMR-Ray \cite{xmrray} proposed more refined network-flow features specific to the Stratum pool mining protocol, and employed one-class classification techniques to allow the model to be trained using only mining traffic.

While these systems report good detection rates and scalability, they are limited to the specific mining pool protocol on which they were trained. The authors of XMR-Ray noted that hand-crafted obfuscations and deviations from the expected protocol can affect the detection rates of these methods. Upgrading or deviating from a mining pool protocol is a much easier task than changing the underlying cryptomining algorithm, which requires a revision to the blockchain protocol, and hence network-based detection methods require less effort to evade.

\noindent\\
\textbf{Resource and Performance Metrics.} Researchers have proposed cryptomining detection based on CPU, memory, and other resource consumption metrics. These methods employ side channels to detect the secondary effects of proof-of-work computations. Wu et al. \cite{app12199838} and Gomes and Correia \cite{9306696} proposed machine learning classification of cryptominers based on CPU usage metrics. DeCrypto Pro \cite{9196224} introduced a more comprehensive set of performance counters for classification by reading processor, memory, and disk metrics. Outguard \cite{10.1145/3308558.3313665} and CoinSpy \cite{kelton2020browser} included network features and information from the JavaScript engine, such as execution time, compilation time, and garbage collection statistics, in addition to raw performance data. Gangwal and Conti \cite{8854845} proposed a magnetic side-channel by profiling magnetic field emission of a processor during cryptomining.

The authors of these studies noted a few drawbacks of using performance metrics and side channels to detect cryptojacking. Namely, these systems are sensitive to noise resulting from external processes running concurrently. They also require administrator privileges to monitor performance counters and system-level events. In addition, miners can restrict their behavior by throttling or performing arbitrary tasks concurrently to manipulate performance metrics. Coinspy \cite{kelton2020browser} and Outguard \cite{10.1145/3308558.3313665} attempt to alleviate these drawbacks by incorporating both network and performance metrics.

\noindent\\
\textbf{Semantic Instruction Count.} The results of Seismic \cite{wang2018seismic} indicate that cryptomining behavior and proof-of-work algorithms can be differentiated at the semantic instruction level. Their findings highlight the significance of a few binary instructions such as \verb|and|, \verb|xor|, and \verb|shr| in cryptominers. Based on this observation, they proposed a detection method based on the statistical distribution of instructions in dynamically collected execution traces. Carlin et al. \cite{8514167} performed a similar analysis by training machine learning models on opcode distribution. MineSweeper \cite{10.1145/3243734.3243858} hand-crafted algorithm-specific signatures using instruction count analysis. MineThrottle \cite{bian2020minethrottle} refined this approach by profiling only a small number of frequently executed blocks of code and checking the distribution of instructions within these blocks.

The detection of cryptominers based on instruction counts can be easily circumvented by inserting spurious operations to skew the distribution of instructions. While MineThrottle \cite{bian2020minethrottle} tries to address this issue by profiling only frequently executed code blocks, blocks can be duplicated to hide their true frequency or divided into smaller blocks to dilute the incriminating instructions. The method we propose in this paper is built on the foundation of semantic instruction distribution, but we also incorporate the structure of data-flow to improve the robustness of our detection method.

\noindent\\
\textbf{Binary File Analysis. } Romano and Wang \cite{romano2020wasim} proposed WASim, a classification method of WebAssembly binaries using features extracted from Wasm binary files such as function sizes, export types, file attributes, and other metadata. The features are used to train several machine learning models for classification. The authors of MINOS \cite{naseem2021minos} discovered that WebAssembly cryptojacking binaries often look similar when represented directly as grayscale images and proposed a convolutional neural network classifier on image representations of the binaries. Although they report high detection rates, this method is not robust because the WebAssembly binary format contains sections with debugging information which can grow arbitrarily large, allowing the image data to be modified arbitrarily. Cabrera-Arteaga et al. \cite{CABRERAARTEAGA2023103296} and Harnes and Morrison \cite{harnes2024cryptic} demonstrated that MINOS is also vulnerable to multiple forms of obfuscation, which is supported by our experimental results.

\noindent\\
The body of prior work indicates that there is yet room for improvement in cryptojacking malware detection that is resilient against obfuscation and evasion.

\subsection{Resource Graphs in Malware Detection}

Resource graphs provide rich information on the behavior of a program, and as such, a large body of work exists on the use of graphs in detecting malicious software. Hu et al. \cite{10.1145/1653662.1653736} designed a system called SMIT which uses approximate graph edit distances of function-call graphs to compute the $K$ nearest neighbors in a malware database. Although they concluded that function-call graphs are less susceptible to obfuscations, many tools have since been developed to perform sophisticated transformations on functions \cite{tigressHome}, \cite{ieeespro2015-JunodRWM}. Kinable et al. \cite{KinableJoris} proposed similar techniques to cluster call-graphs using approximate edit distance and density-based clustering (DBSCAN).

Park et al. \cite{10.1145/1852666.1852716} defined the similarity measure \textit{maximal common subgraph} that is used to classify system call graphs by comparing them to those generated by malicious software. The normalized similarity of two graphs is the ratio between the size of the maximal shared subgraph and the larger of the two graphs. Although this concept of similarity resembles the subgraph similarity we propose in this paper, our measure focuses on smaller fragments of the subgraph to be more resistant to fragmentation and computationally feasible on large graphs. Hisham et al. \cite{8752028} utilized graph algorithmic properties such as density, shortest path, diameter, radius, and other centrality measures to construct features from control flow graphs, which are classified using machine learning. Yamaguchi et al. \cite{6956589} introduced a novel representation of source code called \textit{code property graph}, which merges abstract syntax trees, control flow graphs, and program dependence graphs into a single structure. The combined data structure can be mined effectively to discover vulnerabilities and other properties of the program.

Later works also adopted deep learning techniques to classify resource graphs. For example, Gao et al. \cite{GAO2021102264} demonstrated the effectiveness of graph convolutional networks in detecting android malware based on their API usage graphs. Anderson et al. \cite{andersonblake} presents the only study we know of that uses instruction-level resource graphs to classify malicious software behavior. They use the adjacency of instructions in a dynamic execution trace to construct a Markov chain of assembly instructions, which are classified using graph kernels and machine learning.

The body of existing work on the usage of resource graphs in malware detection fails to address the challenges of detecting obfuscated cryptojacking malware. High-level system-call and API usage graphs do not provide much insight into the numerical computation of cryptominers. Furthermore, function call graphs and control flow graphs are trivially obfuscated using standard obfuscators such as Tigress \cite{tigressHome} and OLLVM \cite{ieeespro2015-JunodRWM}.

At the instruction level, resource graphs grow significantly larger, making them much more challenging to analyze. The difficulty in utilizing instruction traces lies in simplifying massively granular information in a meaningful way and comparing the extracted features effectively. The method proposed by Anderson et al. \cite{andersonblake} to analyze instruction traces captures only the adjacency of instructions rather than the underlying data flow, which could be vulnerable to instruction reordering and speculative execution. Most importantly, the Markov chain representation does not effectively convey local substructures in the data-flow graphs, making it difficult to uncover superimposed computation. The current literature lacks a way to effectively simplify instruction-level data flow graphs and compare them effectively. The techniques we proposed in this paper allow us to exploit these graphs which have never been explored in malware detection.

Finally, machine learning techniques are less suitable than fingerprinting methods due to the limited diversity of cryptominers \cite{tekiner2021browser, 9092245} and mining algorithms \cite{BAJRA2024102571}. The results of Tekiner et al. \cite{tekiner2021browser} suggest that the large number of mining samples used in prior machine learning studies are most likely duplicate cryptominer samples originating from a small number of service providers mining a few select cryptocurrencies. In the remainder of this study, we propose a method of simplifying large graphs while preserving repetitive local features. We also define a novel notion of subgraph similarity that captures local graph properties and is resistant to obfuscation. Our paper represents the first step in comprehensively utilizing instruction-level data-flow graphs in cryptominer detection.

\subsection{WebAssembly}

WebAssembly (Wasm) is a low-level bytecode format designed to run at near-native performance on a wide variety of systems. It aims to complement JavaScript by providing a platform for deploying high-performance software on the web, as well as providing a portable compilation target for higher-level languages such as C and C++ \cite{10.1145/3140587.3062363}. While WebAssembly has seen gradual adoption into the mainstream \cite{cncfCNCFAnnual}, \cite{chandramouli2024data}, recent studies indicate that it is still largely used for illicit purposes such as obfuscation and cryptomining \cite{10.1007/978-3-030-22038-9_2}, \cite{9860829}. WebAssembly presents an exciting future for delivering fast and energy-efficient applications through the web, but more studies need to be conducted on its security implications and mitigation.

A Wasm binary takes the form of a module which contains functions, globals, tables, and memories, that can be exported and imported to integrate with JavaScript environments. WebAssembly operates as a stack machine, meaning that a function consists of a sequence of instructions that manipulate values on an operand stack, popping argument values and pushing results to the stack, as opposed to register machines which perform operations on registers. An example WebAssembly snippet and the corresponding data-flow graph are presented in \autoref{fig:dfgex}. The details of our instrumentation and data-flow collection methods are described in \autoref{section:implementation}.

\begin{figure}
    \begin{center}
    \includegraphics[width=0.5\linewidth]{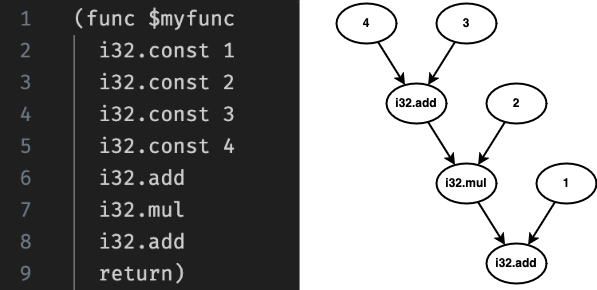}
    \caption{A short WebAssembly function in text format alongside its data-flow graph.}
    \label{fig:dfgex}
    \end{center}
\end{figure}

\section{Threat Model}

In this study, we consider malicious web pages which mine cryptocurrency in the background. We assume that the majority of the proof-of-work algorithm is implemented in WebAssembly, which is the case for most cryptojacking scripts in the wild based on previous studies \cite{kharraz2019outguard}, \cite{10.1007/978-3-030-22038-9_2}, \cite{10.1145/3243734.3243858}. The attacker may mine with or without a pool, and protocol communications may be subject to any obfuscation. The web page may employ throttling or perform arbitrary tasks concurrently, and the WebAssembly binary may be subjected to standard code obfuscation and anti-analysis transformations. We assume that the detector has full access to the browser and is able to collect instruction traces of all WebAssembly execution.

\section{Graph Analysis}

Our analysis is motivated by the observation that proof-of-work algorithms perform extensive amounts of repetitive and artificial computation. In data-flow graphs of dynamic single assigned variables, redundant computations emerge as regular substructures, allowing us to compress the graph into a much smaller form while preserving local semantics by combining these structures until there is little to no repetition remaining. Since the graphs represent computation at a very low abstraction level, we hypothesize that it should also maintain a degree of similarity between equivalent programs, even under obfuscation. First, we introduce our method of simplifying large graphs based on the aforementioned observations, then we define a new subgraph similarity measure capable of discovering similarities between seemingly different graphs of equivalent programs.

The overall approach we take in detecting cryptominers is to maintain a database of cryptominer fingerprints generated from their data-flow graphs. The fingerprints are simplified versions of the original graphs. To test a sample for malicious behavior, we compare the sample's fingerprint with the known cryptominer fingerprints in our database using a subgraph similarity measure to determine whether a cryptominer fingerprint is a subset of the sample's behavior.

\begin{figure}
\captionsetup[subfigure]{justification=Centering}
\begin{subfigure}[t]{0.2\textwidth}
    \includegraphics[width=\textwidth]{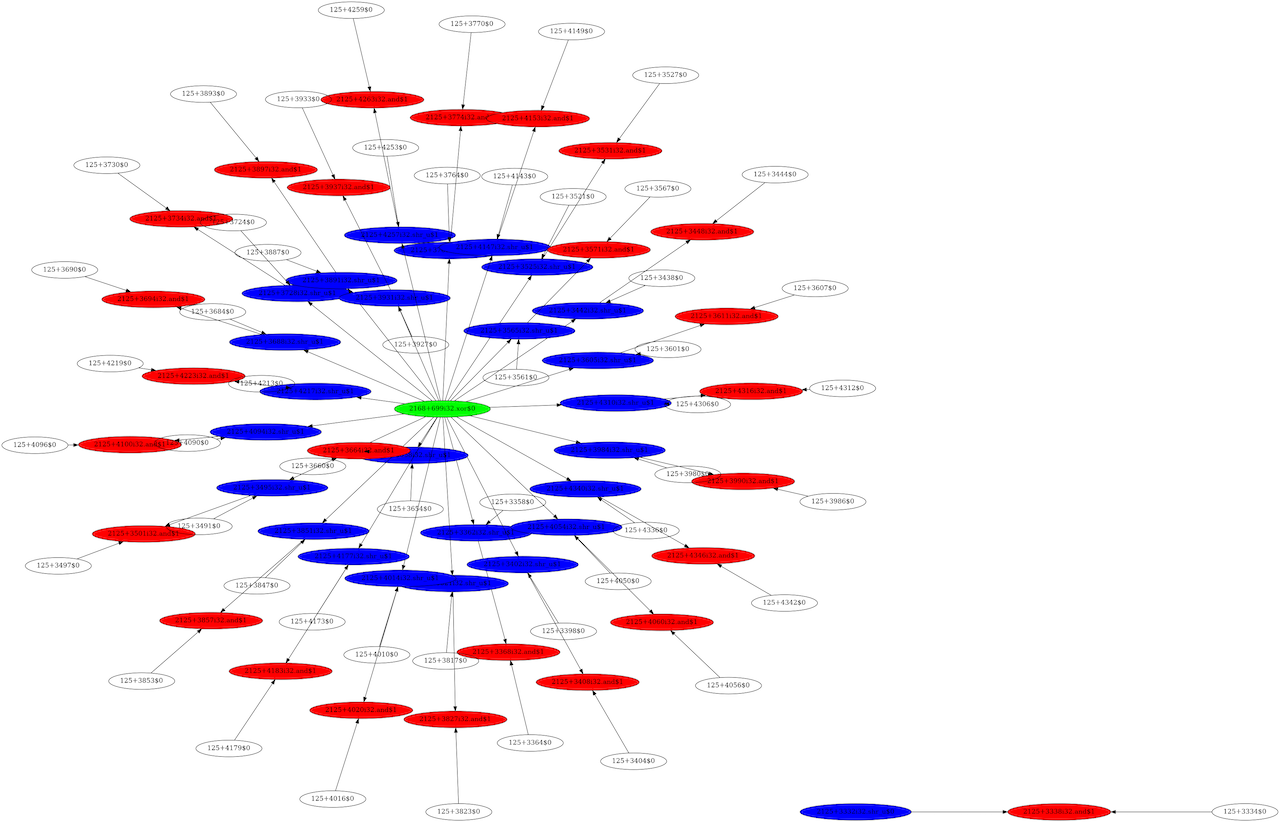}
    \caption{CryptoNight}
    \label{figure:cn}
\end{subfigure}\hspace{\fill}
\begin{subfigure}[t]{0.2\textwidth}
    \includegraphics[width=\linewidth]{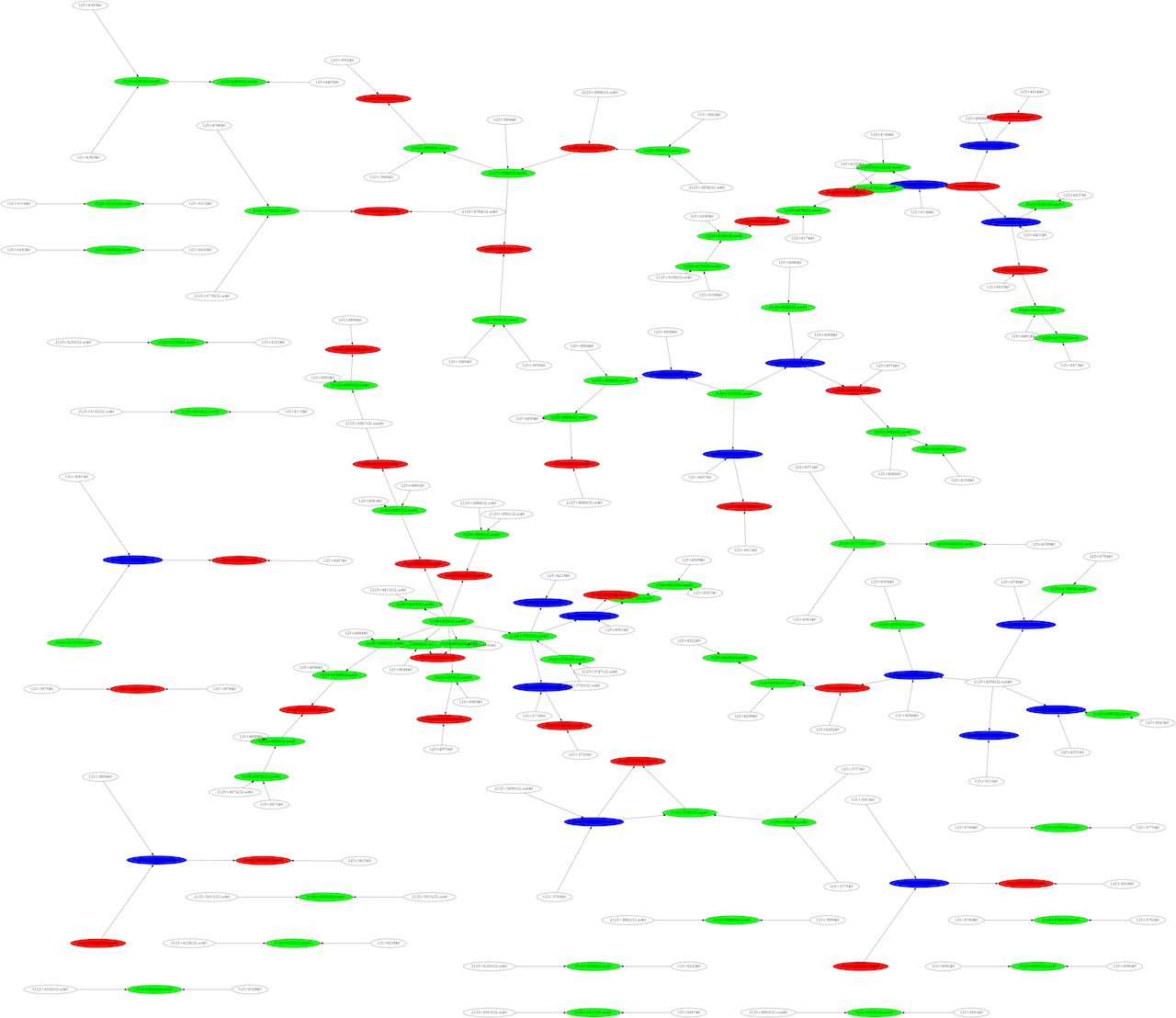}
    \caption{emcc-obf substitution}
    \label{figure:cn-ems}
\end{subfigure}\hspace{\fill}
\begin{subfigure}[t]{0.2\textwidth}
    \includegraphics[width=\linewidth]{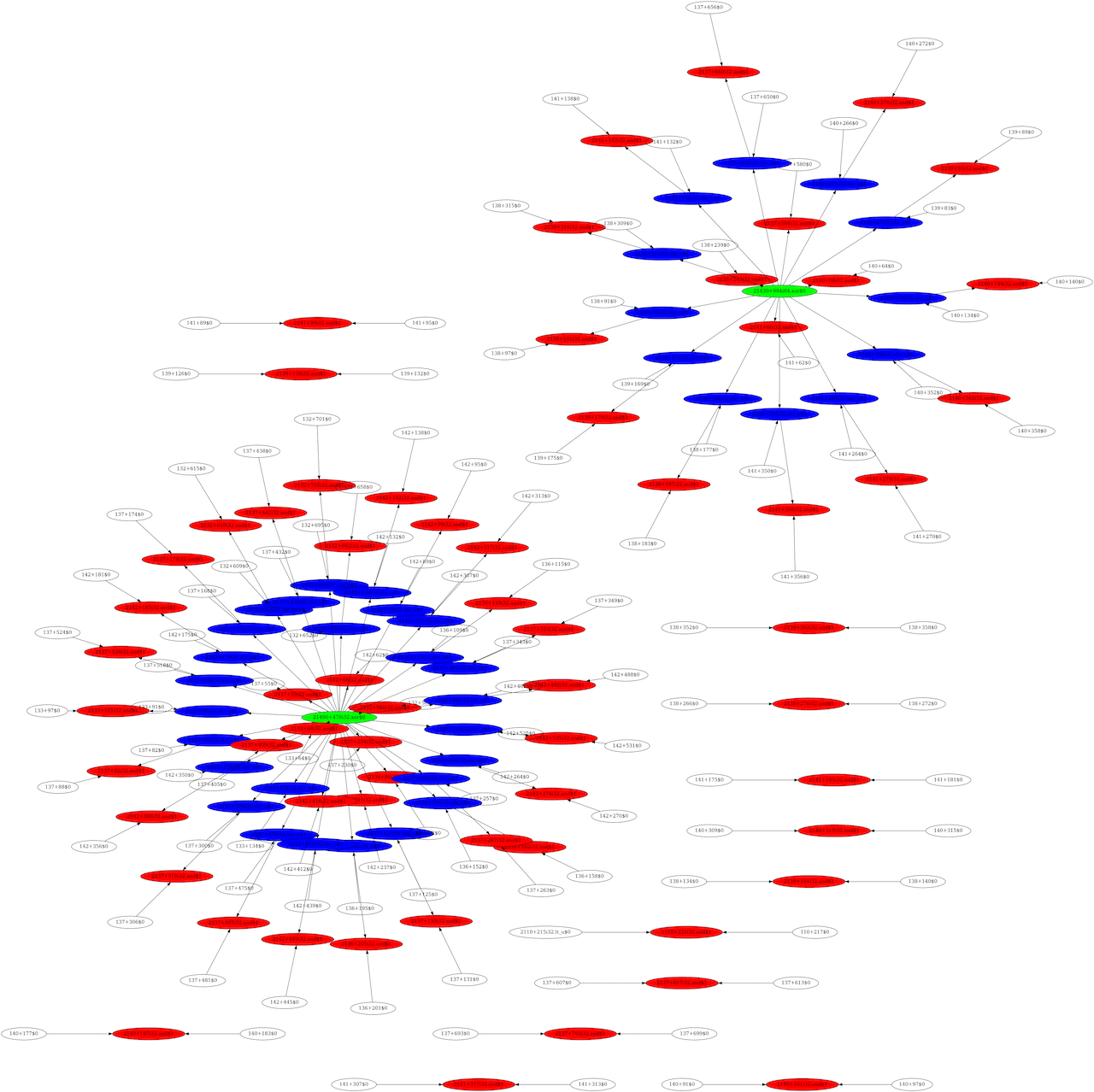}
    \caption{Tigress function split and flatten}
    \label{figure:cn-sf}
\end{subfigure}\hspace{\fill}
\begin{subfigure}[t]{0.2\textwidth}
    \includegraphics[width=\linewidth]{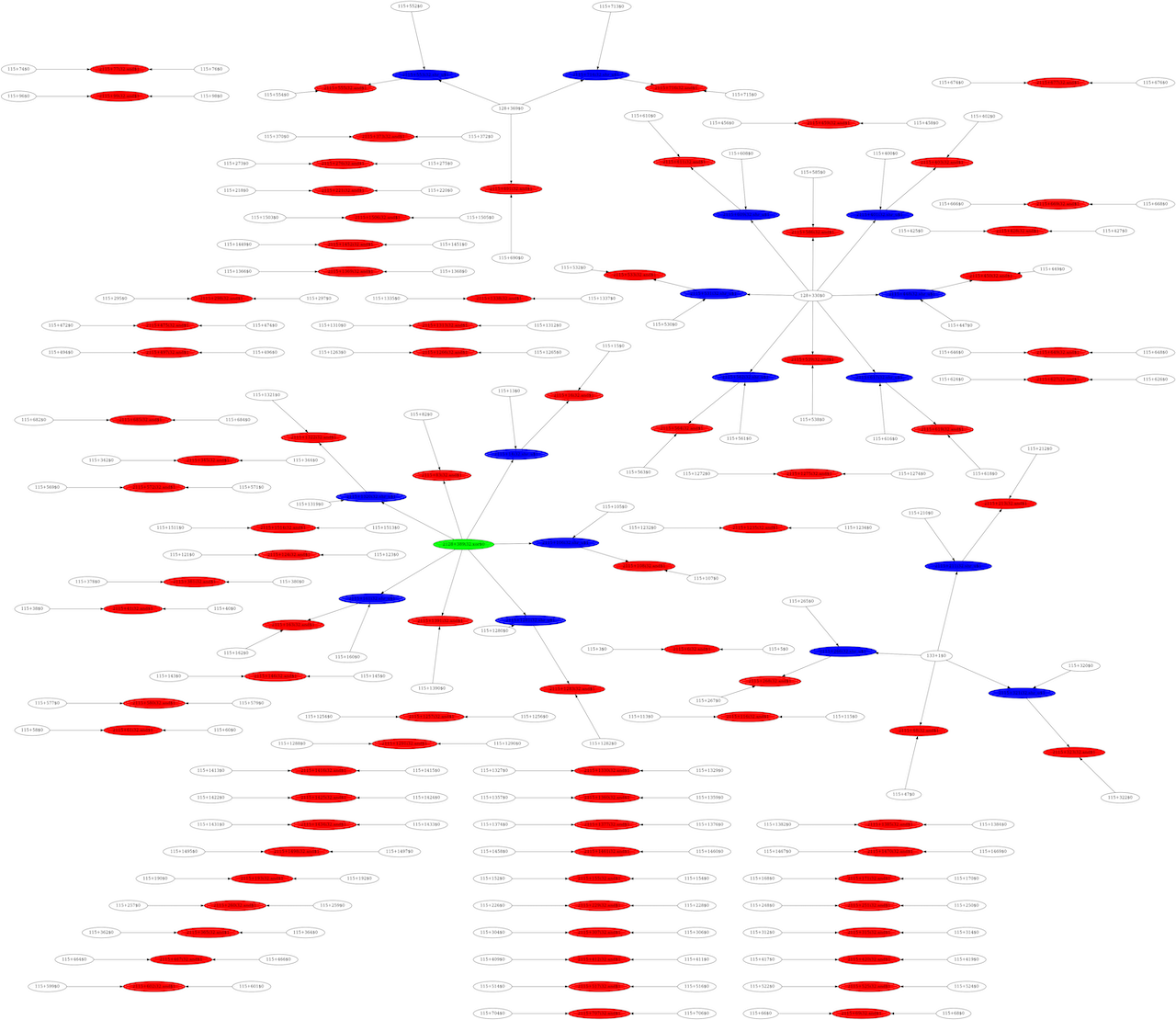}
    \caption{Tigress encode arithmetic}
    \label{figure:cn-ea}
\end{subfigure}
\bigskip
\\
\begin{subfigure}[t]{0.2\textwidth}
    \includegraphics[width=\linewidth]{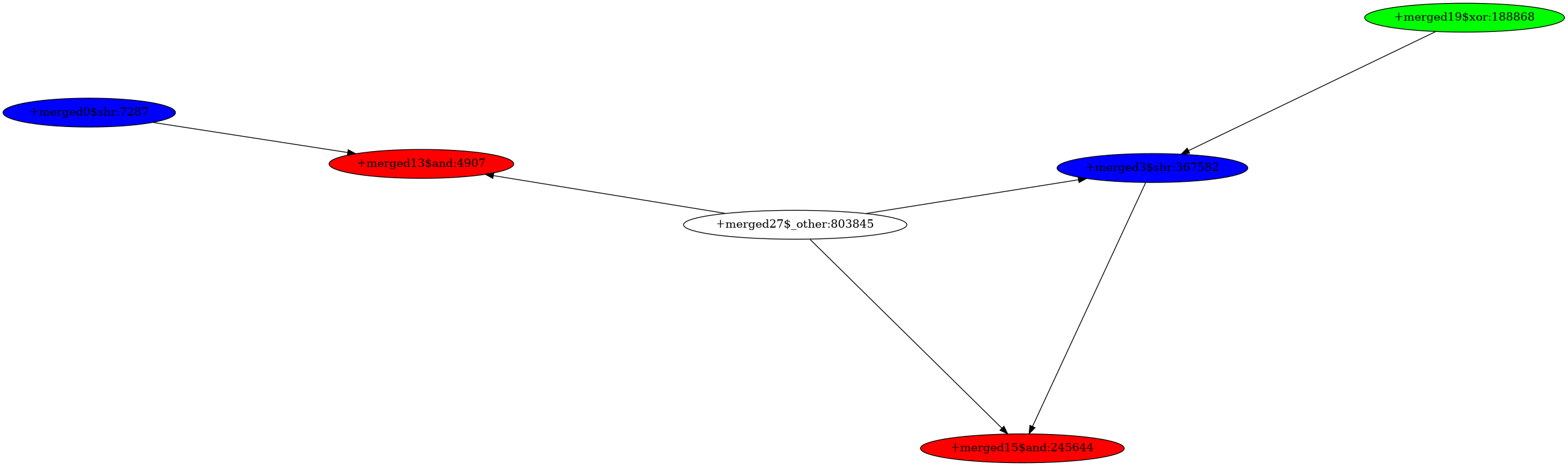}
    \caption{Simplified CryptoNight}
    \label{figure:cn-simple}
\end{subfigure}\hspace{\fill}
\begin{subfigure}[t]{0.2\textwidth}
    \includegraphics[width=\linewidth]{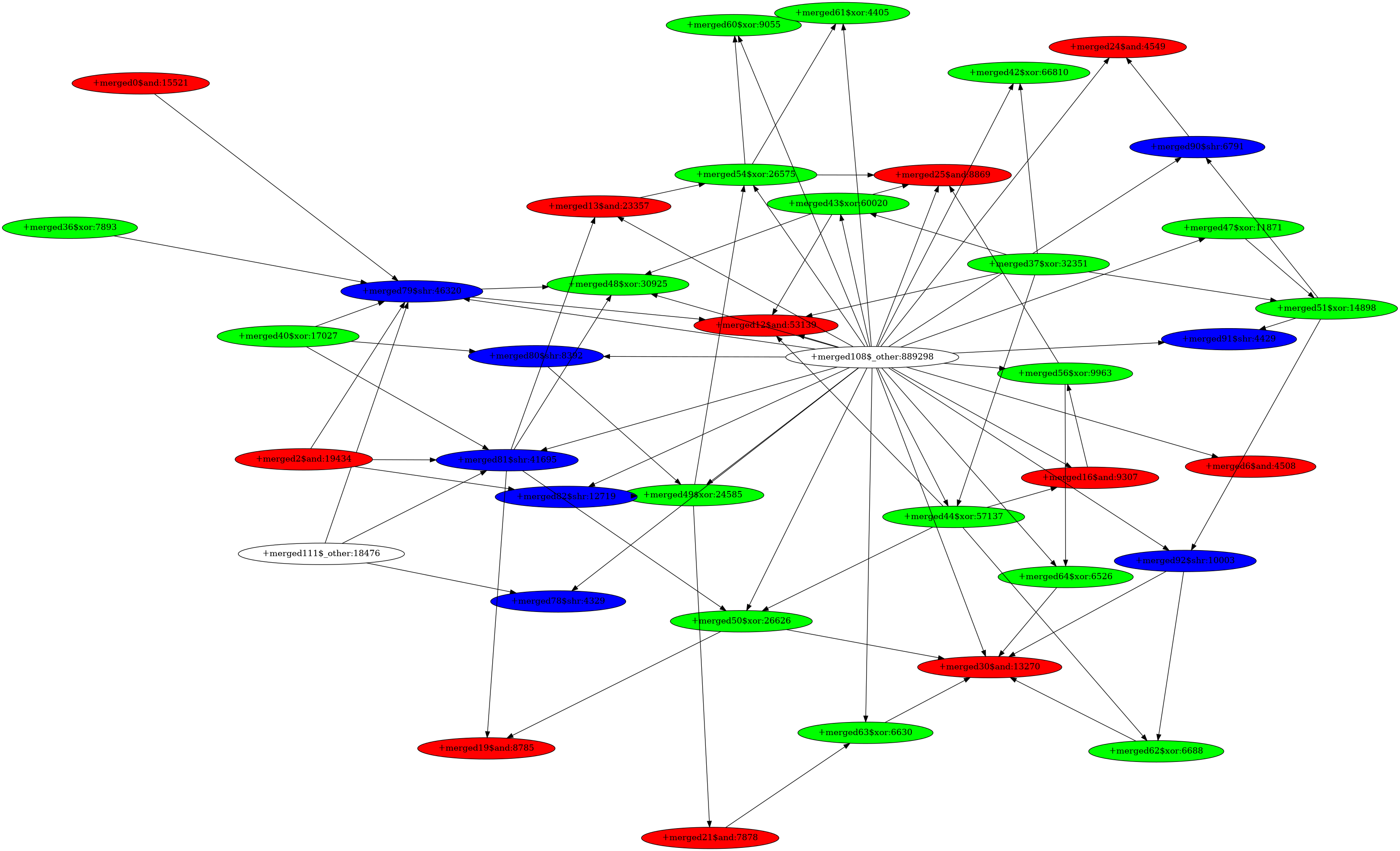}
    \caption{Simplified emcc-obf substitution}
    \label{figure:cn-ems-simple}
\end{subfigure}\hspace{\fill}
\begin{subfigure}[t]{0.2\textwidth}
    \includegraphics[width=\linewidth]{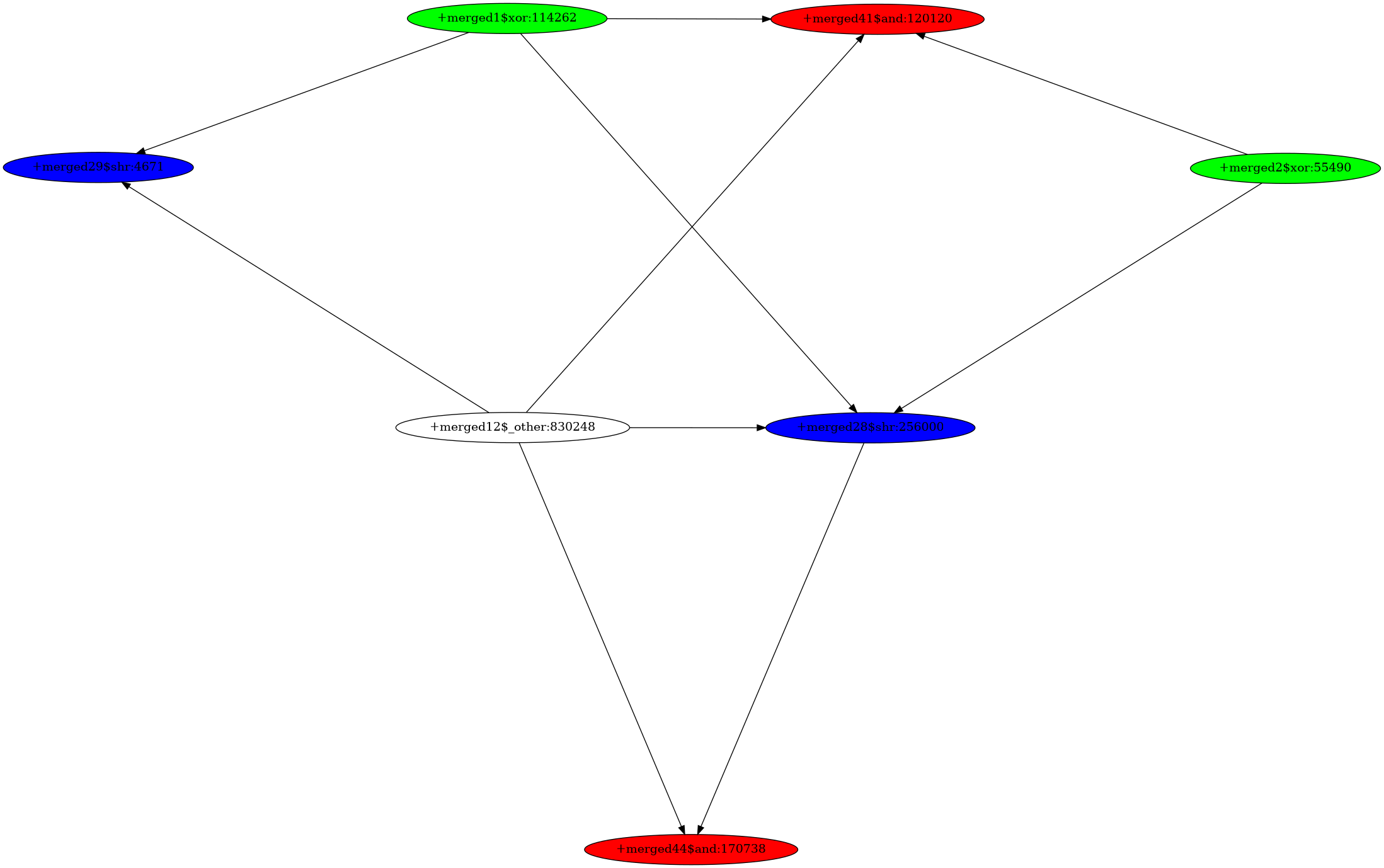}
    \caption{Simplified Tigress function split and flatten}
    \label{figure:cn-sf-simple}
\end{subfigure}\hspace{\fill}
\begin{subfigure}[t]{0.2\textwidth}
    \includegraphics[width=\linewidth]{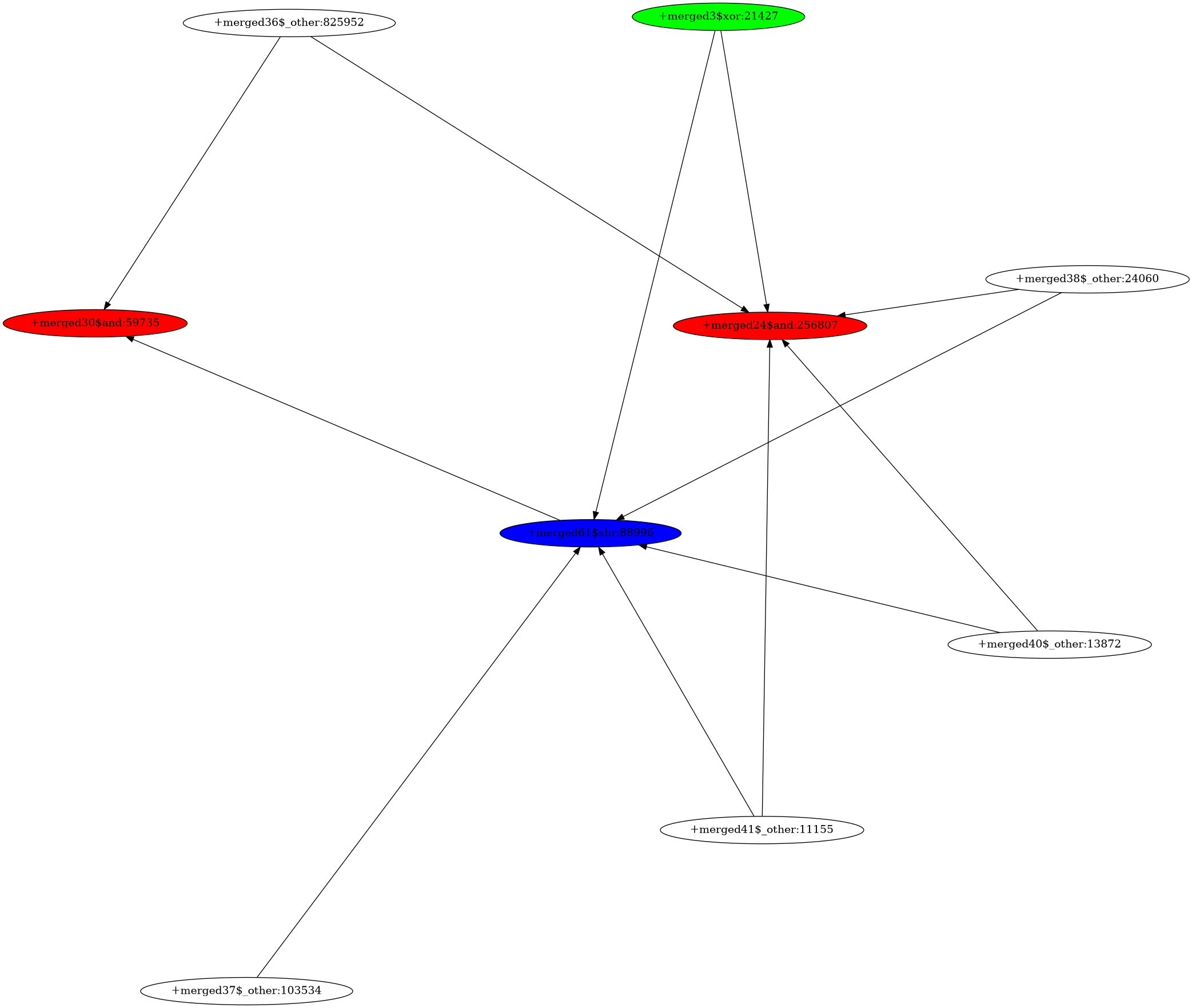}
    \caption{Simplified Tigress encode arithmetic}
    \label{figure:cn-ea-simple}
\end{subfigure}
\caption{\textbf{(CryptoNight data-flow graphs)} The images show visualizations of the data-flow graphs of the CryptoNight POW algorithm in its original and simplified forms and under different obfuscations. Figure (a) shows the original graph, while (b) to (d) show three obfuscated versions. Figures (e) to (h) show the simplified versions of (a) to (d). Vertices highlighted red, green, and blue represent \texttt{and}, \texttt{xor}, and \texttt{shr} instructions respectively. While other instructions are not traced, they may still appear as data origin in the graph represented by the uncolored vertices.}
\label{figure:cn-all}
\end{figure}

\subsection{Instruction Level Data-Flow Graphs in WebAssembly}
By treating each execution of a WebAssembly instruction as a distinct variable and capturing the flow of data between them, we create a directed acyclic graph with no multi-edges where the runtime variables are single-assigned. This graph represents the flow of data between executions of instructions in the program. The directed acyclic property is used throughout our design. To reduce the instrumentation overhead and computation load, we may record only data-flow into certain instructions of interest in the domain of our malware. Based on the results of Wang et al. \cite{wang2018seismic}, the use of three binary instructions, \verb|and|, \verb|shr|, and \verb|xor|, is characteristic of cryptomining behavior. We elect to instrument only these three instructions to reduce computational burden in comparing graphs.

\subsection{Fingerprinting Through Graph Simplification}
The data-flow graph we collect is necessarily large in order to capture a complete representation of a program's behavior. \autoref{figure:cn} shows a 100 instruction snapshot of the recorded data flow from a CryptoNight mining algorithm. This snapshot considers the three instructions of interest that are executed within the time frame and records the data flow into them. Although the graph is massive, we recognize repeated patterns in the graph which correspond to the iterations of the mining algorithm. In order to generate a fingerprint of a program's data flow, we propose a graph simplification technique that exploits these repetitions to create a compact signature that captures the instruction-level behavior of a program. The key idea in our approach is to merge isomorphic substructures located at the same depth within the graph until we have a minimal representation of the data flow. First, we formally define the repeated substructures, and then we introduce a random walk-based approximation to efficiently compute the simplification.

\begin{definition}[Rooted Subgraph]
    Let $G = (V,E)$ be a directed acyclic graph. A subgraph $S \subseteq G$ is a \textbf{rooted subgraph} if there exists $v \in V(S)$ such that every vertex in $S$ is reachable from $v$. This $v$ is unique when $G$ is acyclic and is called the \textbf{root} of the subgraph. $S$ is \textbf{maximal} if it is the largest subgraph with root $v$.
\end{definition}

The isomorphic \textit{maximal rooted subgraphs} represent repeated substructures that we want to eliminate in the simplified graphs, but isomorphic subgraphs might appear at different locations in the program that are semantically different. To preserve this distinction, we introduce the notion of \textit{depth} which describes where a rooted subgraph $S$ is located within the graph $G$ and refrain from merging subgraphs located at different depths. Note that the \textit{depth} of a rooted subgraph is an external property of $S$ inside the graph $G$ and is not related to its construction.

\begin{definition}[Depth]
    The \textbf{depth} of a vertex $v \in V(G)$ is the number of edges on the longest path in $G$ that ends in $v$. The depth of a rooted subgraph in $G$ is the depth of the root vertex of the subgraph in $G$.
\end{definition}

Intuitively, the depth of a rooted subgraph contains information about the location and sequence of the corresponding instruction in the program that should be preserved in the simplified graph. \autoref{figure:simplifyprocess1} shows many maximal isomorphic rooted subgraphs located at the same depth in the CryptoNight algorithm with the roots highlighted. These subgraphs represent the same computation carried out in different iterations of the algorithm. Finally, we iteratively merge all maximal isomorphic rooted subgraphs with distinct roots of the same depth until we reach a fixed point. Note that there is effectively no difference between merging the entire subgraph and only merging the roots since we merge until reaching a fixed point. \autoref{figure:simplifyprocess} shows this process on a smaller version of the CryptoNight graph. The graph shown in \autoref{figure:simplifyprocess1} shows the center and a few branches of the large circular structure seen in the original CryptoNight graph in \autoref{figure:cn}. \autoref{figure:simplifyprocess5} is related to \autoref{figure:cn-simple} but not exactly the same due to the approximation method we introduce in the next section.

\begin{figure}
\captionsetup[subfigure]{justification=Centering}
\begin{subfigure}[t]{0.18\textwidth}
    \includegraphics[width=\textwidth]{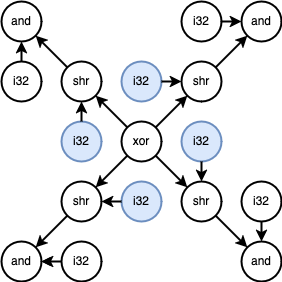}
    \caption{}
    \label{figure:simplifyprocess1}
\end{subfigure}\hspace{\fill}
\begin{subfigure}[t]{0.18\textwidth}
    \includegraphics[width=\linewidth]{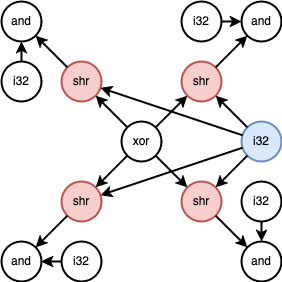}
    \caption{}
    \label{figure:simplifyprocess2}
\end{subfigure}\hspace{\fill}
\begin{subfigure}[t]{0.18\textwidth}
    \includegraphics[width=\linewidth]{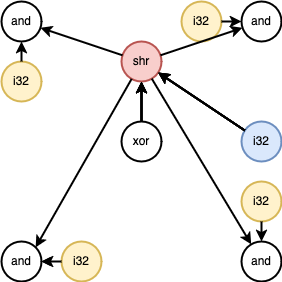}
    \caption{}
    \label{figure:simplifyprocess3}
\end{subfigure}\hspace{\fill}
\begin{subfigure}[t]{0.18\textwidth}
    \includegraphics[width=\linewidth]{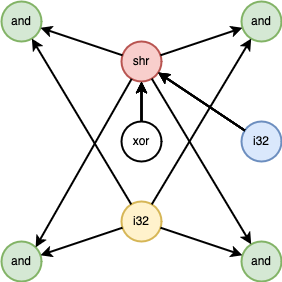}
    \caption{}
    \label{figure:simplifyprocess4}
\end{subfigure}\hspace{\fill}
\begin{subfigure}[t]{0.18\textwidth}
    \includegraphics[width=\linewidth]{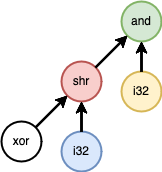}
    \caption{}
    \label{figure:simplifyprocess5}
\end{subfigure}
\caption{A step-by-step reduction of the CryptoNight graph. The root vertices of isomorphic subgraphs are highlight and merged at each step.}
\label{figure:simplifyprocess}
\end{figure}

\subsubsection{Approximation Through Backward Random Walks}

The process of simplifying the graph requires us to search for isomorphic subgraphs in large data-flow graphs. While a naive algorithm may perform exact subgraph matching, it would be intractable because the subgraph isomorphism problem is NP-complete \cite{DBLP:books/acm/23/Cook23a}. For reference, the graphs in our experiment have up to 1000 nodes and 2000 edges. To reduce the complexity of the simplification process, we introduce an approximate algorithm that simplifies graphs using backward random walks. Intuitively, when walking backwards randomly on a graph starting from a random vertex, the probability of a specific vertex being visited is largely dependent on the structure and size of its descendants. By computing the probability of a random backward walk visit, we get an approximate characterization of the maximal rooted subgraph of a vertex.

\begin{definition}[Backward Random Walk]
    A \textbf{backward walk} is a sequence of vertices $$P = \{v_1, v_2, ..., v_k | v_i \in V(G) \text{ for } 1 \le i \le k\}$$ such that there exists an edge $e = v_{i+1} \rightarrow v_{i}$ for $1 \le i \le k - 1$, and $v_k$ has no incoming edge. The backward walk visits vertices in the opposite direction of the edges (i.e. $v_1, v_2, ..., v_k$).
\end{definition}

Since the graph is acyclic, each vertex is visited at most once per backward walk. Denote $P(v)$ the probability that a vertex $v \in V(G)$ is visited in a backward random walk. The following lemma establishes a connection between the backward random walk and the maximal rooted subgraphs.

\begin{lemma}\label{lemma:probability}
    Let $G$ be a directed acyclic graph with no multi-edges. The probability that a vertex $v \in V(G)$ is visited in a random backward walk is $$P(v) = \frac{1}{|V(G)|} + \sum_{v_c \in C(v)}\frac{1}{|I(v_c)|}P(v_c)$$ where $C(v)$ denotes the set of children vertices of $v$ in the directed graph and $|I(v_c)|$ denotes the number of incoming edges into $v_c$.
\end{lemma}

The intuition behind Lemma 4 is that, given that a backward walk $W$ contains $v$, either: (1) $v$ is the first vertex in the backward walk; or (2) the backward walk visits a child of $v$ then proceeds (backward) to $v$ itself. Therefore, the probability of visiting $v$ can be decomposed to the probability that its children will be visited.

\begin{proof}
    Consider a backward random walk $W$. Let $E_0$ denote the event that $v$ is the first vertex in the backward walk $W$. Denote $v_{c_1}, v_{c_2}, ..., v_{c_M} \in C(v)$ the children vertices of $v$. Let $E_i$ for $1 \le i \le M$ denote the event that the backward walk visits $v_{c_i}$ and then $v$ consecutively. Then the event that $v \in W$ is 
    $$E(v \in W) = E_0 \cup \bigcup_{i=1}^{M} E_i $$
    In other words, $v$ is in $W$ if and only if $W$ starts with $v$, or $W$ visits a child of $v$ and proceeds to $v$. It is clear that $E_0$ is mutually exclusive to every $E_i$ where $i \neq 0$. Moreover, every $E_i$ for $1 \le i \le M$ is mutually exclusive, since if two $E_i$ and $E_j$ are true for $i \neq j$, the graph would contain a cycle. So that 
    $$
        P(v \in W) = P(E_0) + \sum_{i=1}^{M} P(E_i)
        = \frac{1}{|V(G)|} + \sum_{v_c \in C(v)}\frac{1}{|I(v_c)|}P(v_c)
    $$
\end{proof}

\autoref{lemma:probability} implies that the probability that a vertex is visited in a random backward walk is characterized entirely by its maximal rooted graph. We state this formally in the next theorem.

\begin{theorem}\label{theorem:isomorphic}
    Let $H_1, H_2 \subseteq G$ be maximal isomorphic rooted subgraphs such that every pair of isomorphic vertices in $H_1$ and $H_2$ contains the same number of incoming edges in $G$. Then the roots $v_1 \in H_1$ and $v_2 \in H_2$ have the same probability of being visited in a random backward walk.
\end{theorem}

To prove this theorem, we need the following short lemma.

\begin{lemma}\label{lemma:longestpath}
    Every path of the longest length in a maximal rooted subgraph $H \subseteq G$ contains the root vertex.
\end{lemma}

\begin{proof}
    Let $P = \{v_1, v_2, ..., v_n\}$ be an arbitrary path in $H$ such that the root vertex $v$ is not in $P$. Since $v$ is the root, there exists a path $P'$ from $v$ to $v_1$. Since $H$ is acyclic, $P'$ cannot intersect $P$, otherwise we can form a cycle. Therefore, we can form a new path by concatenating $P'$ and $P$ which is longer than $P$ and contains $v$, proving the lemma.
\end{proof}

\begin{proof}[Proof of Theorem 5]
    We proceed by induction on the length $N$ of the longest path in $H_1$ and $H_2$. 
    \begin{description}
        \item[Base:]
            When $N=0$, $H_1$ and $H_2$ are singleton graphs.
            Thus $$P(v_1) = P(v_2) = \frac{1}{|V(G)|}$$
        \item[Induction:]
            Since $H_1$ and $H_2$ are isomorphic, the sets $C(v_1)$ and $C(v_2)$, the children of the corresponding root vertices, are also isomorphic. Consider two isomorphic children $c_{1} \in C(v_1)$ and $c_{2} \in C(v_2)$. Each of these two vertices induce a maximal rooted subgraph $H'_1$ and $H'_2$ which are also isomorphic. Moreover, $H'_1$ and $H'_2$ are subgraphs of $H_1$ and $H_2$ which do not contain the roots $v_1$ and $v_2$, thus, the longest path in $H'_1$ and $H'_2$ is at most $N-1$ by \autoref{lemma:longestpath}. By the induction hypothesis, $P(c_1) = P(c_2)$ for all isomorphic pair, $c_1 \in C(v_1)$ and $c_2 \in C(v_2)$. Finally, since we assume that every pair of isomorphic vertices in $H_1$ and $H_2$ contains the same number of incoming edges in $G$, $|I(c_1)| = |I(c_2)|$ as well for such pairs. Therefore,
            $$P(v_1) = 
            \frac{1}{|V(G)|} + \sum_{c_1 \in C(v_1)}\frac{1}{|I(c_1)|}P(c_1) =
            \frac{1}{|V(G)|} + \sum_{c_2 \in C(v_2)}\frac{1}{|I(c_2)|}P(c_2) =
            P(v_2)$$
    \end{description}
\end{proof}

\paragraph*{Approximate Simplification Algorithm}

\autoref{theorem:isomorphic} tells us that much of the structural information of the subgraphs are embedded in the roots in a backward random walk. Although it is possible that distinct non-isomorphic maximal rooted subgraphs may coincide with the same root vertex probability, or that isomorphic subgraphs may result in different root probabilities if some of the vertices have different numbers of incoming edges, we expect this to happen rarely in the type of computation graphs that we are working with. It is justifiable by \autoref{theorem:isomorphic} that we can approximately merge all maximal isomorphic rooted subgraphs by combining all vertices with the same random backward walk probabilities. As it does not matter whether we merge the entire subgraph or just the root, the process is straightforward. \autoref{algo:simplifyapprox} details the overall process. \autoref{figure:cn-simple} shows the result of applying \autoref{algo:simplifyapprox} to the CryptoNight graph. More examples of simplified graphs are shown in \autoref{figure:simplifiedgraphsexample}.

\begin{figure}[t]
\captionsetup[subfigure]{justification=Centering}
\begin{subfigure}[t]{0.15\textwidth}
    \includegraphics[width=\textwidth]{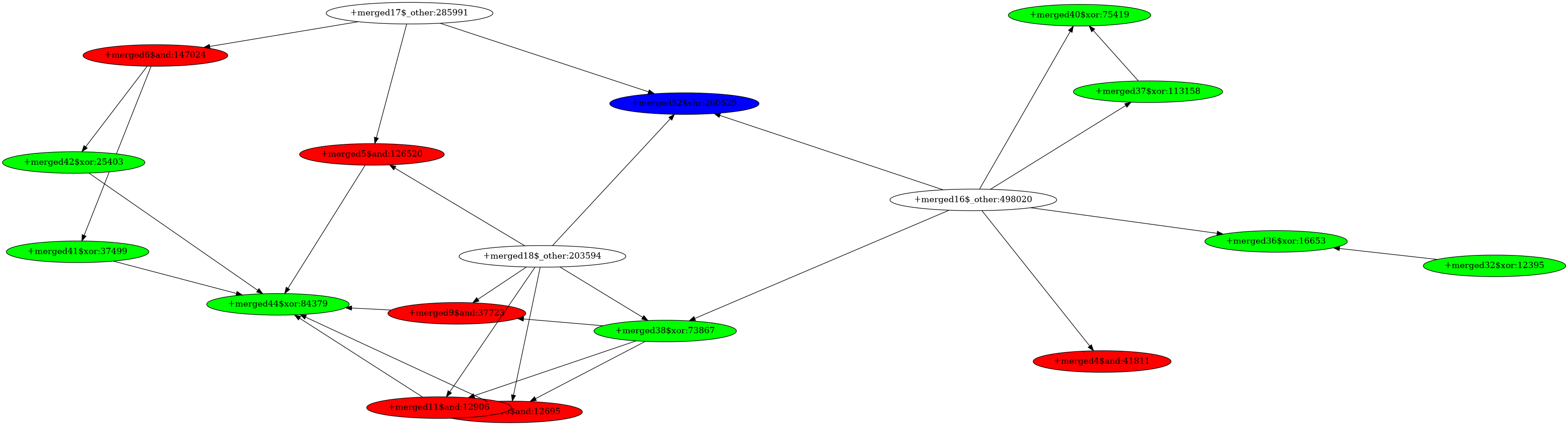}
    \caption{btc}
    \label{figure:btc}
\end{subfigure}\hspace{\fill}
\begin{subfigure}[t]{0.15\textwidth}
    \includegraphics[width=\linewidth]{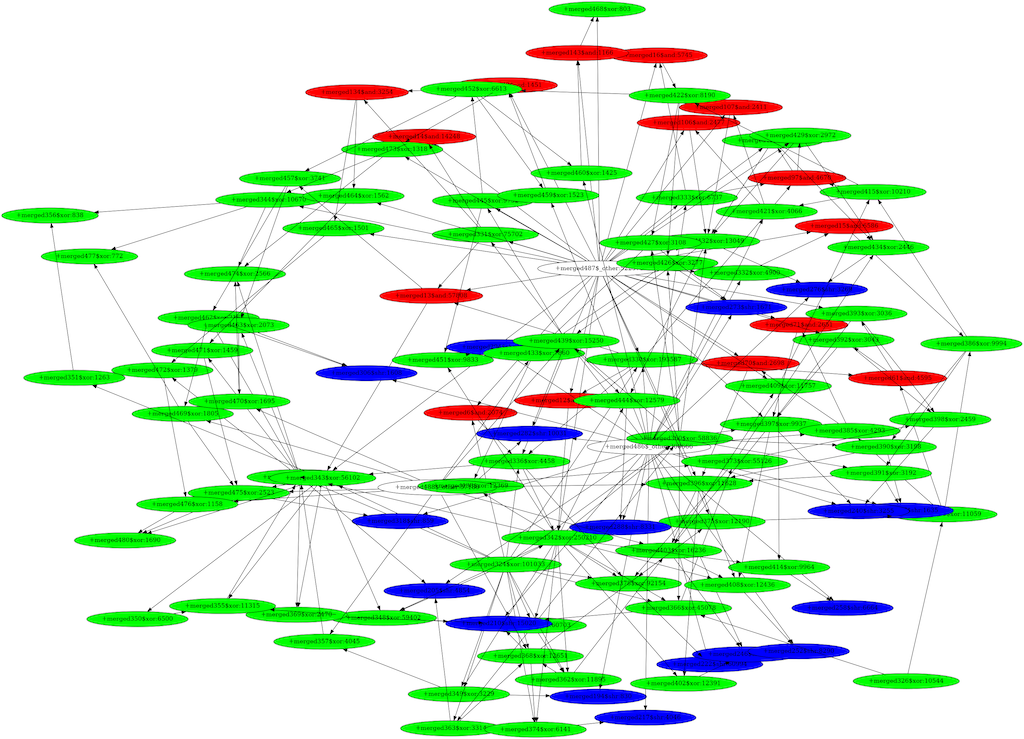}
    \caption{eth}
    \label{figure:eth}
\end{subfigure}\hspace{\fill}
\begin{subfigure}[t]{0.15\textwidth}
    \includegraphics[width=\linewidth]{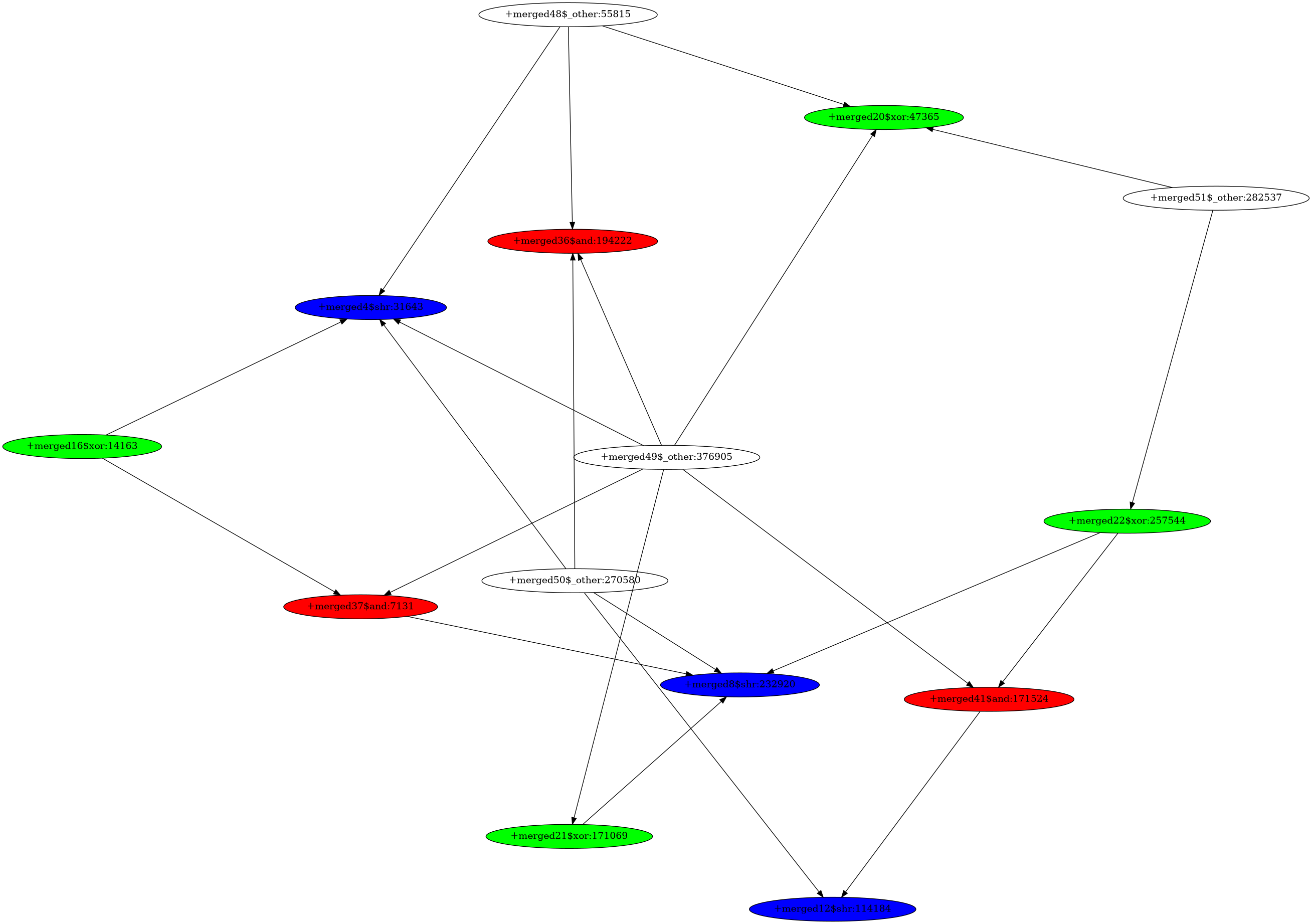}
    \caption{zny}
    \label{figure:zny}
\end{subfigure}\hspace{\fill}
\begin{subfigure}[t]{0.15\textwidth}
    \includegraphics[width=\linewidth]{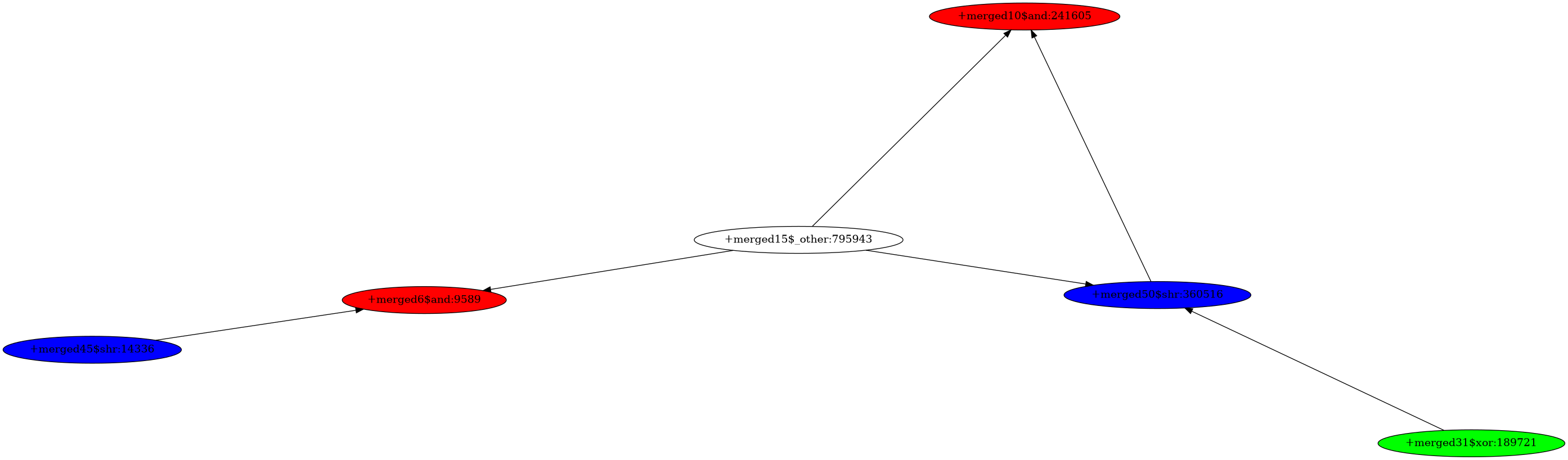}
    \caption{cn}
    \label{figure:cn2}
\end{subfigure}\hspace{\fill}
\begin{subfigure}[t]{0.15\textwidth}
    \includegraphics[width=\linewidth]{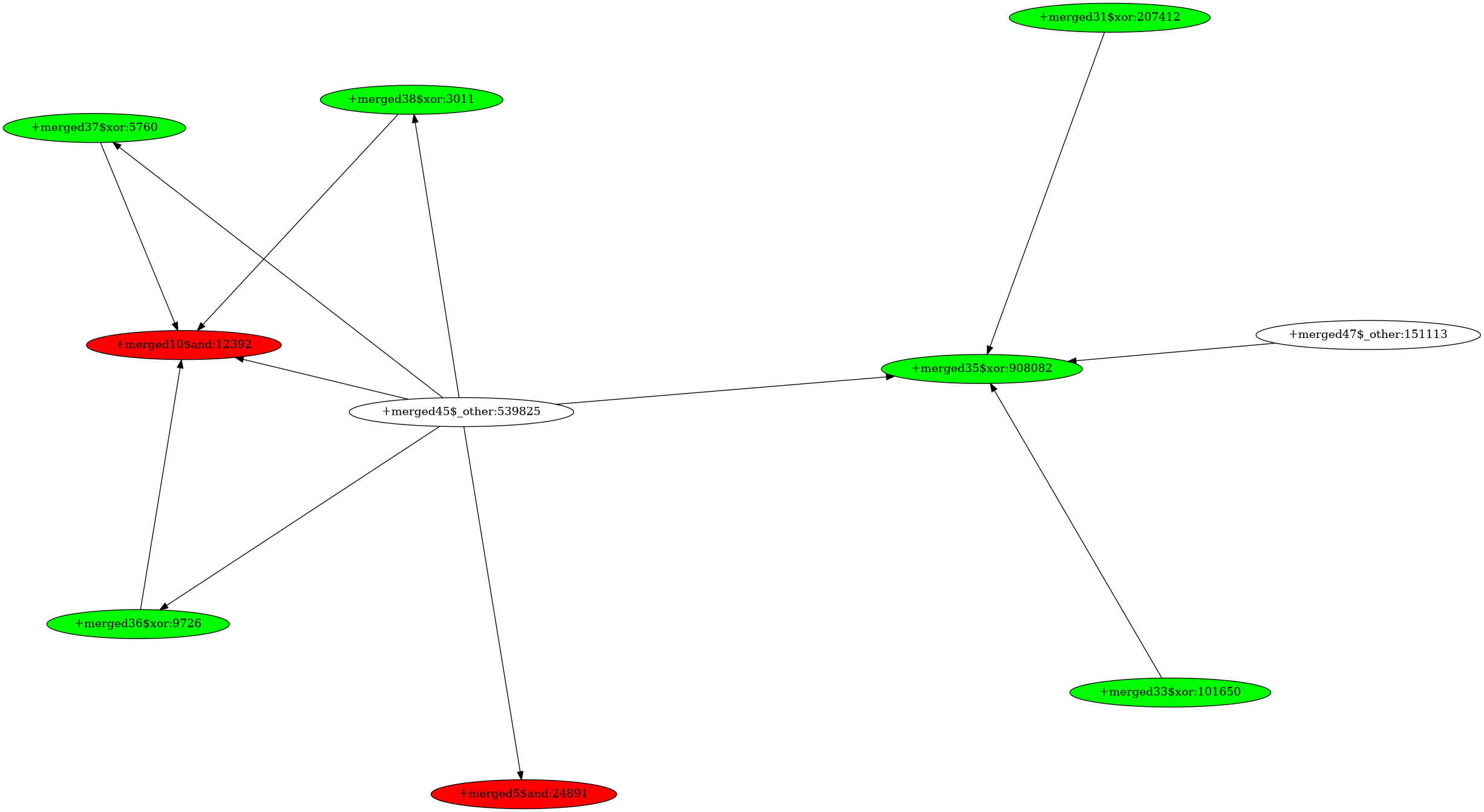}
    \caption{wmp}
    \label{figure:wmp}
\end{subfigure}\hspace{\fill}
\begin{subfigure}[t]{0.15\textwidth}
    \includegraphics[width=\linewidth]{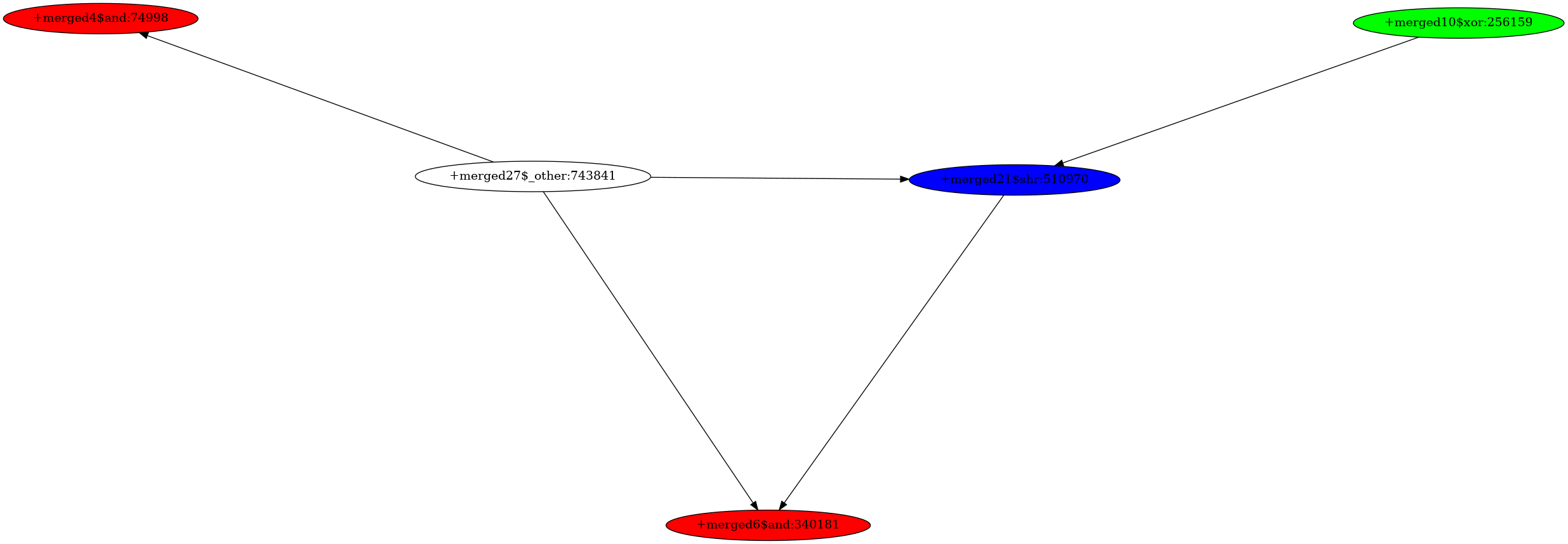}
    \caption{xmr}
    \label{figure:xmr}
\end{subfigure}
\bigskip
\\
\begin{subfigure}[t]{0.15\textwidth}
    \includegraphics[width=\linewidth]{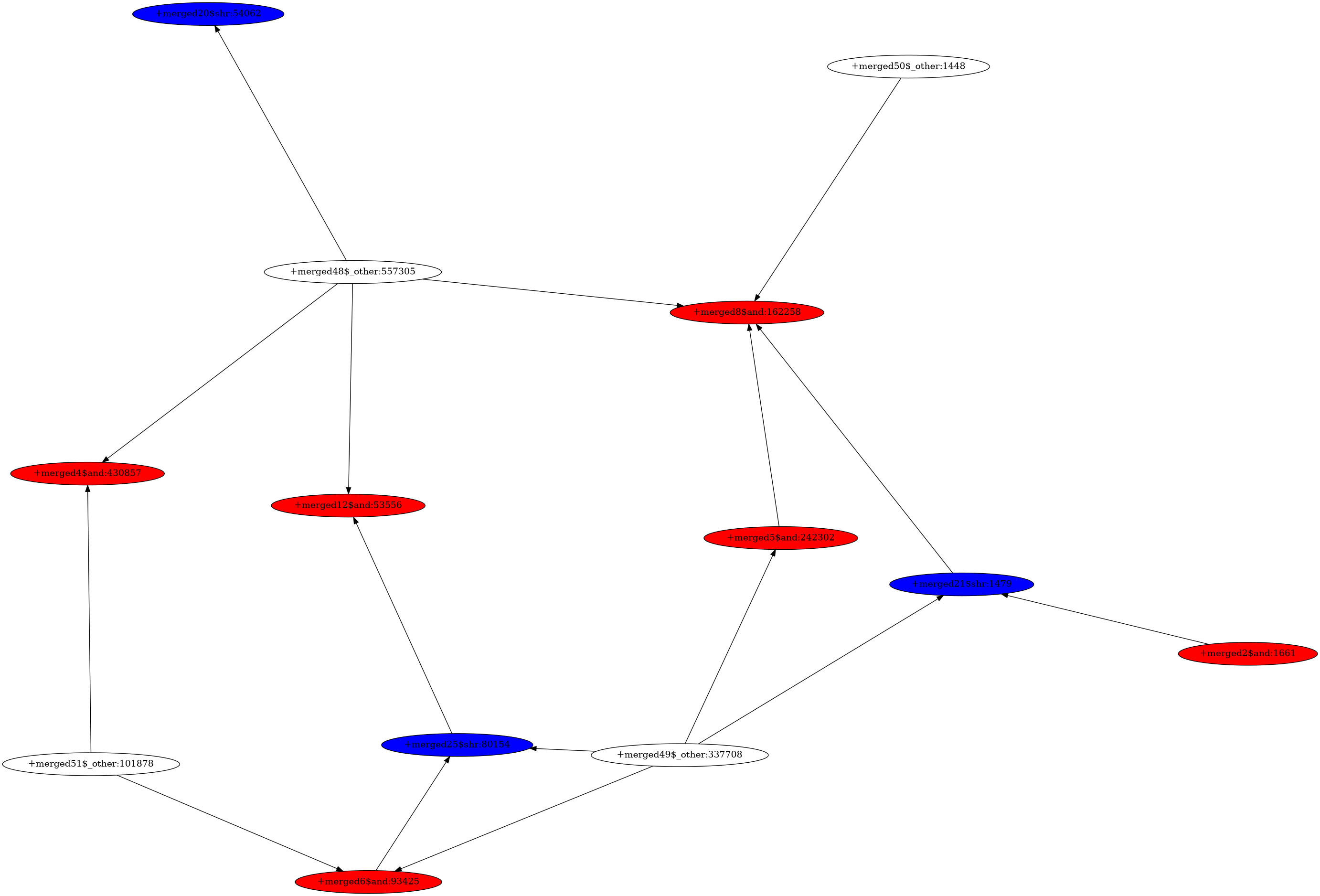}
    \caption{boa}
    \label{figure:boa}
\end{subfigure}\hspace{\fill}
\begin{subfigure}[t]{0.15\textwidth}
    \includegraphics[width=\linewidth]{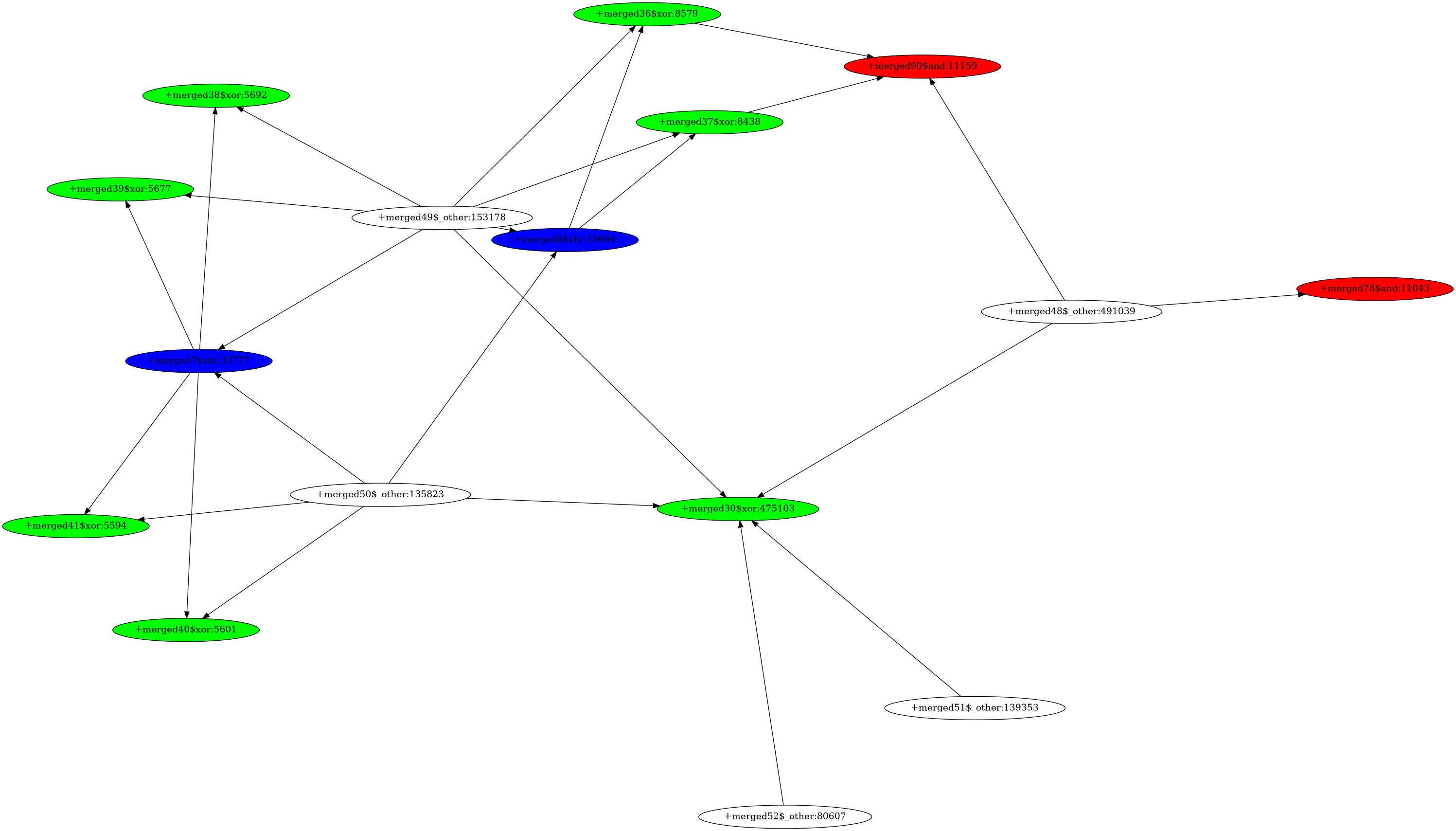}
    \caption{bullet}
    \label{figure:bullet}
\end{subfigure}\hspace{\fill}
\begin{subfigure}[t]{0.15\textwidth}
    \includegraphics[width=\linewidth]{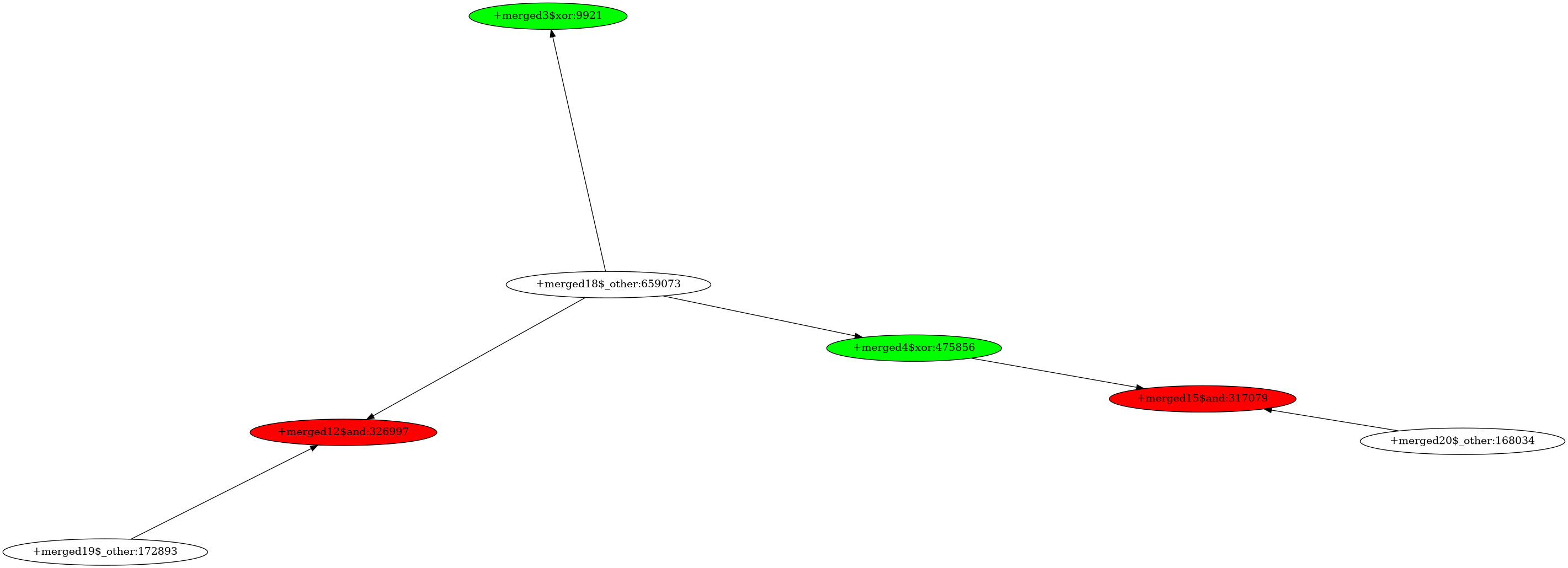}
    \caption{chocolatekeen}
    \label{figure:chocolatekeen}
\end{subfigure}\hspace{\fill}
\begin{subfigure}[t]{0.15\textwidth}
    \includegraphics[width=\linewidth]{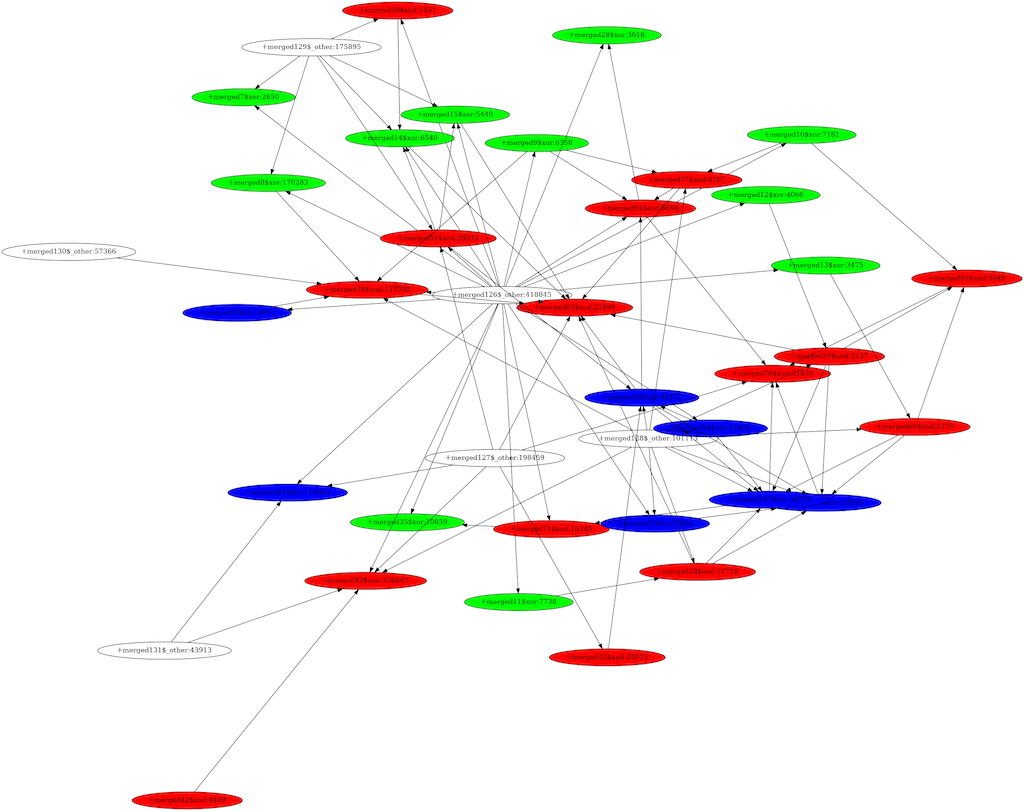}
    \caption{ffmpeg}
    \label{figure:ffmpeg}
\end{subfigure}\hspace{\fill}
\begin{subfigure}[t]{0.15\textwidth}
    \includegraphics[width=\linewidth]{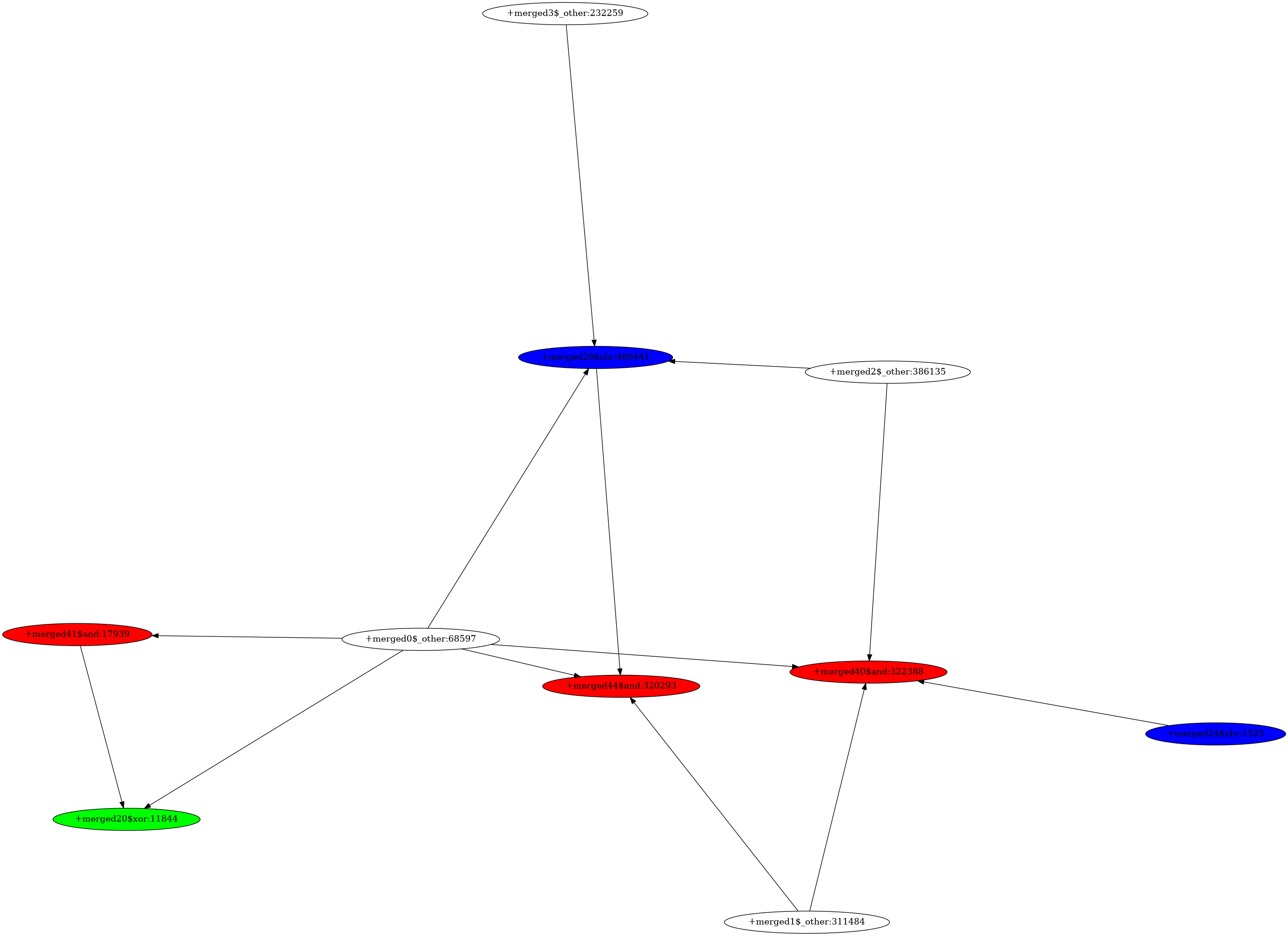}
    \caption{sandspiel}
    \label{figure:sandspiel}
\end{subfigure}\hspace{\fill}
\begin{subfigure}[t]{0.15\textwidth}
    \includegraphics[width=\linewidth]{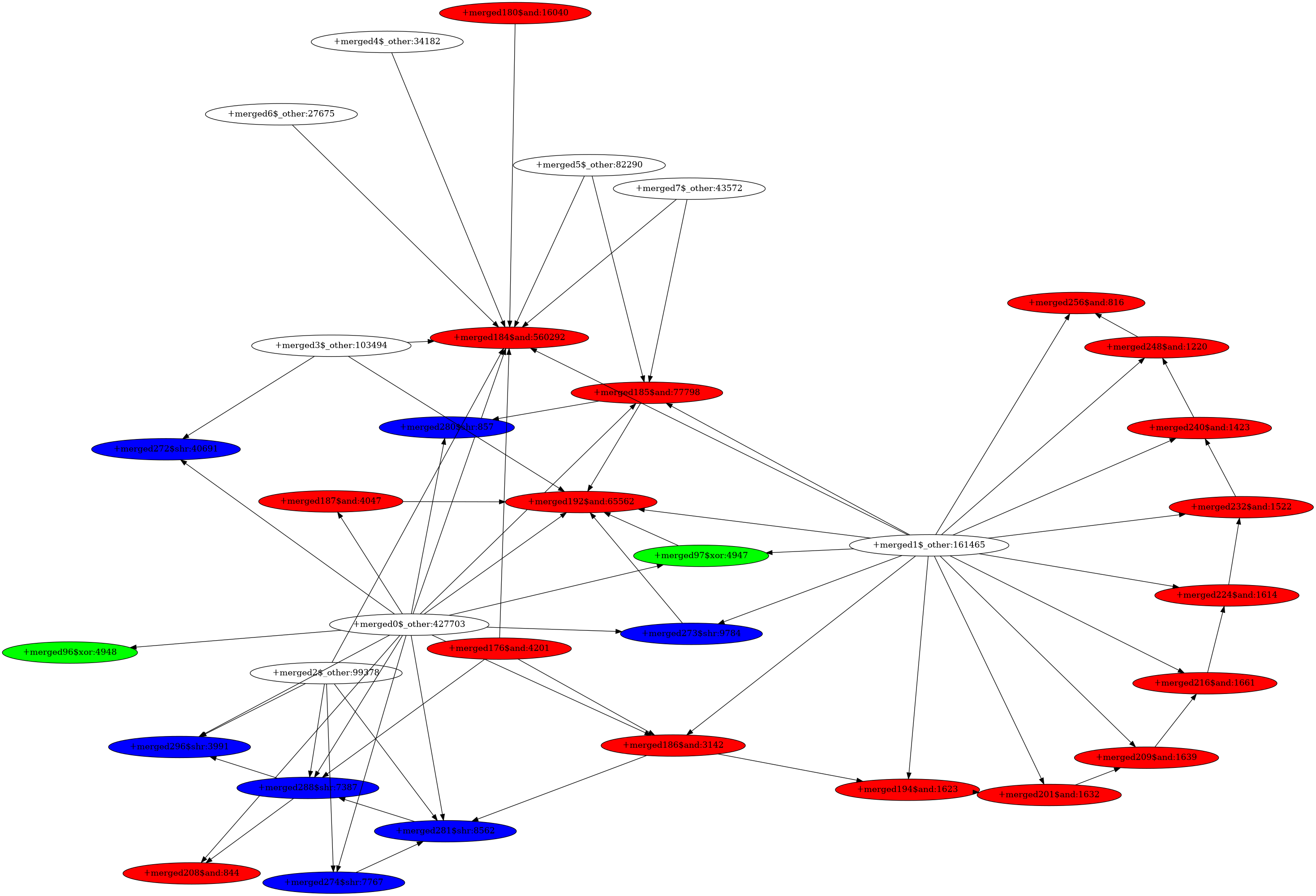}
    \caption{sqlgui}
    \label{figure:sqlgui}
\end{subfigure}
\caption{\textbf{(Simplified Graphs of Miners and Non-miners)} The images show the simplified graphs of the cryptominers in the first row and 6 real-world web applications in the second row. The details of these samples are listed in \autoref{table:minersdesc} and \autoref{table:nonminersdesc}. Vertices highlighted red, green, and blue represent \texttt{and}, \texttt{xor}, and \texttt{shr} instructions respectively. While other instructions are not traced, they may still appear as data origin in the graph represented by the uncolored vertices. The corresponding original graphs can be found in \autoref{sec:appendixunsimpex}.}
\label{figure:simplifiedgraphsexample}
\end{figure}

\paragraph*{Quality of Approximation} Although it is difficult to give a general bound for the approximation, we can show the tightness of the bound for small graphs by constructing and enumerating all approximately simplified and exactly simplified graphs. The construction is given by the following definition.

\begin{definition}\label{defineconstruct}
    Let $G_N$ be a directed acyclic graph with maximum vertex depth $N$. Denote $S_i$ the set of all vertices of depth $N-i$ in $G_N$. We construct $G_N$ as follows:
    \begin{itemize}
        \item $S_0$ has a single vertex which contains no outgoing edges.
        \item $S_i$ contains vertices which point to any combination of vertices with higher depth, or equivalently vertices in $\bigcup_{j = 0}^{i - 1} S_j$. It follows that $|S_i| = |P(\bigcup_{j = 0}^{i - 1} S_j)|$ where $P$ denotes the power set. The power set $P(\bigcup_{j = 0}^{i - 1} S_j)$ are the sets of children of each element of $S_i$.
    \end{itemize}
\end{definition}

\autoref{fig:construction} shows the construction of $G_3$. Note that $|S_3| = 2048$ and $|S_4| = 2^{2059}$. The next theorems show that $G_N$ contains all approximately simplified graphs and exactly simplified graphs of maximum depth $N$.

\begin{theorem}\label{theoremsimplifysubseteq}
    Let $H$ be a directed acyclic (exact) simplified graph of maximum depth $N$. Then $H \subseteq G_N$.
\end{theorem}

\begin{proof}
    Consider a directed acyclic simplified graph $H$ with maximum vertex depth $N$. Note that all vertices of depth $n$ can only have children of depth at least $n+1$. It follows that at depth $N$, every rooted subgraph is a singleton, thus there can only be one vertex of depth $N$ since $H$ is simplified. Therefore $S_0$ has a single vertex which contains no outgoing edges, satisfying the first condition of $G_N$. Moreover, since every vertex can only have children of greater depth, the second condition of $G_N$ is satisfied. Thus, $H \subseteq G_N$
\end{proof}

\begin{theorem}
    Let $H$ be a directed acyclic approximately simplified graph of maximum depth $N$. Then $H \subseteq G_N$.
\end{theorem}

\begin{proof}
    The proof is mostly the same as that of \autoref{theoremsimplifysubseteq}. The only difference is that all vertices of depth $N$ must have the same backward random walk probabilities since they have no outgoing edges, thus they must all be merged, and so $|S_0| = 1$, proving the theorem.
\end{proof}

Denote $\mathbb{H}_4$ the set of all subgraphs $H \subseteq G_4$ with $|V(H)| \le 80$. Denote $\mathbb{A}_4$ the set of all approximately simplified graphs $A$ with depth at most $4$ and $|V(A)| \le 80$. Denote $\mathbb{E}_4$ the set of all (exact) simplified graphs $E$ with depth at most $4$ and $|V(E)| \le 80$. By the two theorems above, $\mathbb{A}_4 \subseteq \mathbb{H}_4$ and $\mathbb{E}_4 \subseteq \mathbb{H}_4$. Using \autoref{defineconstruct}, we can construct the elements of $\mathbb{H}_4$. By random sampling $\mathbb{H}_4$, we can show that about $90\%$ of all $H \in \mathbb{H}_4$ are also in $\mathbb{A}_4$. Each element of $\mathbb{H}_4$ also contains an average of $6.7$ isomorphisms between maximal rooted subgraphs. These numbers are computed by inspecting a large random sample of $\mathbb{H}_4$. Intuitively, this means that, for small graphs, the approximate simplification is close to the exact simplification. Conversely, since $90\%$ of all $H \in \mathbb{H}_4$ are in $\mathbb{A}_4$, and $\mathbb{E}_4 \subseteq \mathbb{H}_4$, it is highly likely most exact simplified graphs are fixed points for the approximate algorithm, implying that the approximate simplifying algorithm will arrive at a non-trivial solution. These bounds are relevant since almost all graphs in our experiment satisfy the condition of having depth at most $4$ and at most $80$ vertices.

\begin{figure}
    \centering
    \includegraphics[width=0.5\linewidth]{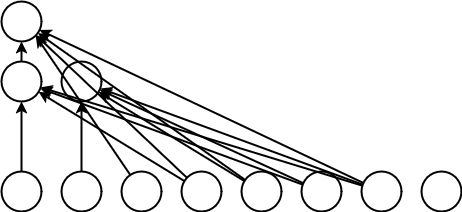}
    \caption{Visualization of $G_3$ in \autoref{defineconstruct}}
    \label{fig:construction}
\end{figure}

\paragraph*{Grouping Vertices} While the true probability of visiting a vertex in a random backward walk can be computed, we chose to give an approximation by performing the random backward walk a large number of times to tolerate noises in the graph. Since this approximation contains slight variations from the true value, we need a method to group vertices which are close in probability. Suppose that we execute the random backward walk $N$ times, the number of times a vertex $v$ is visited is a binomial distribution $M_v \sim B(N,P(v))$. Therefore, we can use known binomial confidence intervals to cluster the frequencies given by the random backward walk. In our implementation, we chose to use simple mean shift clustering \cite{scikitlearnMeanShift} on the frequencies of the vertices to achieve a similar result.

\begin{algorithm}
\caption{Approximate Graph Simplification}\label{algo:simplifyapprox}
\hspace*{\algorithmicindent} \textbf{Input} \texttt{Large graph G}\\
\hspace*{\algorithmicindent} \textbf{Output} \texttt{Simplified graph G}
\begin{algorithmic}[1]
\State \texttt{Do random backward walks to approximate the probability P(v)}
\For{\texttt{d = 0 to max depth in G}}
    \State \texttt{S $\gets$ All vertices of depth d in G}
    \State \texttt{C = \{$S_1, S_2, ..., S_n$\} $\gets$ Cluster S by P(v) using mean shift clustering}
    \State \texttt{G $\gets$ merge all vertices of the same label within the same cluster}
\EndFor
\end{algorithmic}
\end{algorithm}

\subsection{Comparing Graph Fingerprints}

Once we have generated a database of malicious graph fingerprints, we proceed to check whether a sample program contains malicious behavior. Equivalently, the sample program's fingerprint must be checked for traces of any malicious subgraphs. To achieve this, we require a subgraph similarity measure with the following characteristics:
\begin{itemize}
    \item Able to identify the existence of shared local behavior corresponding to short sections of computation.
    \item Tolerates fragmentation and noise in the target super-graph which may appear naturally or due to obfuscation.
    \item Performs well on medium to large graphs.
\end{itemize}
In practical scenarios, there are many sources which may introduce noises to the data-flow graph that we collect. Since we capture only a brief slice of the program, two traces of the same program may capture slightly different parts of the computation. More importantly, various obfuscation techniques can change or remove instructions and flow control, dramatically changing certain parts of the graph. \autoref{figure:cn-all} shows the graphs of the CryptoNight mining algorithm along with three of its obfuscated versions. Large disparities can be seen in the graphs due to insertion, deletion, and fragmentation as a result of the transformations applied to the cryptominer code. The disparities suggest that global similarity measures such as \textit{graph edit distance} are unsuitable for comparing instruction-level data-flow graphs. The ability to tolerate noisy graphs is an important consideration in constructing a robust similarity measure.

\subsubsection{Localized Fragment Similarity}

A natural way to check whether a graph contains the behavior described by another graph is to look for the existence of small substructures of one graph in the other. Intuitively, we are comparing many small fragments in both graphs to search for similar segments. This idea has been employed successfully in data mining and graph databases through indexing methods based on small graph fragments \cite{DBLP:conf/sigmod/ShangLZYW10}, \cite{DBLP:conf/icde/WilliamsHW07}, \cite{DBLP:journals/bioinformatics/TianMSSP07}. Inspecting small fragments gives us a local view of graph similarity, while being tolerant of fragmentation and noise. Moreover, searching for small subgraphs in a large graph can be done efficiently through approximations \cite{DBLP:journals/datamine/BonniciGMMSPG24}, \cite{DBLP:journals/bmcbi/BonniciGPSF13}, \cite{DBLP:journals/tcbb/BonniciG17}. Based on this, we define a simple \textit{n-fragment inclusion score} as a subgraph similarity measure.

\begin{definition}[n-fragment inclusion score]
    The \textbf{n-fragment inclusion score} of a graph $H$ in $G$, denoted $\text{n-FIS}_G(H)$, is the probability that an arbitrary connected subgraph $H' \subseteq H$ containing exactly $n$ edges is a subgraph of $G$.
\end{definition}

\begin{algorithm}[t]
\caption{Approximate n-FIS$_G(H)$}\label{algo:approxnfis}
\hspace*{\algorithmicindent} \textbf{Input} \texttt{H, G, n, k}\\
\hspace*{\algorithmicindent} \textbf{Output} \texttt{approximate n-FIS$_G$(H)}
\begin{algorithmic}[1]
\State \texttt{score $\gets$ 0}
\For{\texttt{i = 1 to k}}
    \State \texttt{S $\gets$ random connected subgraph of H with n edges}
    \If{\texttt{S $\subseteq$ G}}
        \State \texttt{score $\gets$ score + 1}
    \EndIf
\EndFor
\State \texttt{return score/k}
\end{algorithmic}
\end{algorithm}

\paragraph*{Approximation} In order to both avoid enumerating larger graphs and tolerate faulty results in subgraph matching, we approximate the n-FIS score by randomly testing the inclusion of a large number $k$ of connected n-edge subgraphs. \autoref{algo:approxnfis} describes an approximate n-FIS score in detail. For our detection algorithm, we chose based on evaluation $n = 5$ as a compromise between locality and robustness, and $k = 500$ as a good stopping point. In order to test for subgraph inclusion, we use an open source approximate subgraph matching tool, ArcMatch \cite{DBLP:journals/datamine/BonniciGMMSPG24}.

Finally, we test whether a sample is malicious by computing its n-FIS score against the malicious fingerprint database. Given a fingerprint graph $G$ of a sample, we compute the approximate n-FIS$_G(H)$ for each malicious fingerprint $H$ in the database. The existence of at least one score surpassing a threshold implies that the sample contains malicious behavior.

\section{Implementation} \label{section:implementation}

Our implementation of {\ournameshort} focuses on software running through WebAssembly. To generate data-flow graphs, we instrument programs using the open-source WebAssembly dynamic analysis tool Wasabi \cite{DBLP:conf/asplos/LehmannP19}. We use the Wasabi taint analysis framework to generate traces of data flow during execution and filter only three instructions of interest, \verb|and|, \verb|shr|, and \verb|xor|.

The Wasabi analysis framework allows us to insert hooks into the programs and provide call-back functions for every executed instruction. We use these call-backs to record a trace of relevant instructions, their operands, and their results. To trace the data flow of a WebAssembly program which operates as a stack machine, we maintain a shadow stack to track the origin of each operand. Each time an instruction pushes onto the operand stack, the shadow stack stores which instruction the value originated from. Using the information tracked in the shadow stack, we are able to log the flow of data between instructions as edges between the operands and result of each instruction. A visualization of an example execution trace is shown in \autoref{fig:dfgex}. The data-flow trace is collected through the debug console as a list of data-flow edges, and the vertices in the trace are identified by the instruction that pushed the respective value onto the stack. We process this trace into the dynamic single assignment form by separating each variable into distinct generations each time it is written. Finally, we output the graph as a list of edges.

Although {\ournameshort} only instruments three most relevant instructions to cryptomining, we note that the analysis and collection of data-flow traces can be extended to more instructions, which may be relevant when considering other platforms that are more complex than WebAssembly.

{\ournameshort} performs the following process to detect cryptomining. Given pre-processed data-flow graphs $G$ of a sample and $H$ of a known miner, they are simplified into fingerprint graphs $G'$ and $H'$ using \autoref{algo:simplifyapprox}. Next, we compute $5$-FIS$_{G'}(H')$ using \autoref{algo:approxnfis} with $k = 500$ iterations. Scores are classified as malicious or benign on the basis of a static threshold. We present the reported scores and the empirical decision boundary in the evaluation section.

\paragraph*{External Libraries}

We use the Wasabi framework to instrument and perform taint analysis on WebAssembly binaries \cite{DBLP:conf/asplos/LehmannP19}. We use the scikit-learn library to compute the mean shift clusters in \autoref{algo:simplifyapprox}. The approximate subgraph matching in \autoref{algo:approxnfis} is performed using the ArcMatch library \cite{DBLP:journals/datamine/BonniciGMMSPG24}.

\section{Evaluation}

In this section, we evaluate the following research questions.
\begin{description}
    \item[RQ1] How effective is \autoref{algo:simplifyapprox} in simplifying large graphs?
    \item[RQ2] How effective is {\ournameshort} in identifying cryptominers in the presence of obfuscation?
\end{description}

\subsection{Experimental Setup}

We ran our experiments on an 8 core AMD Ryzen 6900HX processor with 16gb memory. To evaluate our research questions, we collected a WebAssembly dataset consisting of 29 real-world web applications, 6 open-source cryptominers, and generated 30 obfuscated cryptominer samples through different obfuscations. We instrument Wasm binaries with a fork of the Wasabi framework developed by the authors of Wasm-R3 \cite{DBLP:journals/pacmpl/BaekGS0TRP24} at the commit version \verb|6836ccd|.

\paragraph*{Cryptominers}

We collected a sample of six open source \texttt{C} and \texttt{C++} miners from GitHub by searching for keywords using GitHub's search engine, e.g. ``wasm'', ``webassembly'', and ``cryptominer''. These searches returned 6 results after filtering out those which do not provide \texttt{C} or \texttt{C++} source code and those which do not compile. We could not use binary samples in previous works due to the lack of functional boilerplate code to execute the WebAssembly. Moreover, the obfuscators that we use in the experiment require \texttt{C} and \texttt{C++} source code for miners, which are detailed in the next section. The summary of these miners are detailed in \autoref{table:minersdesc}.

\begin{table}[t]
\caption{List of WebAssembly cryptominer samples used as evaluation targets.}
\centering
\scriptsize
\begin{tabular}{llll}
\hline
\textbf{Name} & \textbf{URL} & \textbf{Currency} & \textbf{Algorithm} \\ \hline
btc  & \url{https://github.com/kinshukdua/cryptominer}           & Bitcoin               & sha2            \\
eth  & \url{https://github.com/Rachel-Hu/wasm-miner}             & Ethereum              & ethash          \\
zny  & \url{https://github.com/ohac/cpuminer}                    & Bitzeny               & yescrypt        \\
cn   & \url{https://github.com/andrehrferreira/cryptonight-hash} & any CryptoNight coins & CryptoNight v1  \\
wmp  & \url{https://github.com/notgiven688/webminerpool}         & any CryptoNight coins & CryptoNight v4  \\
xmr  & \url{https://github.com/jtgrassie/xmr-wasm}               & Monero                & CryptoNight v1 \\ \hline
\end{tabular}
\label{table:minersdesc}
\end{table}

\begin{table}[t]
\centering
\caption{List of real-world non-miner WebAssembly web applications used in the evaluation.}
\begin{adjustbox}{max width=\textwidth}
\begin{tabular}{lll}
\hline
\textbf{Name} & \textbf{URL} & \textbf{Domain} \\ \hline
boa            & \url{https://boajs.dev/boa/playground}                                                & Programming  \\
bullet         & \url{https://magnum.graphics/showcase/bullet}                                         & Simulator    \\
chocolatekeen  & \url{https://www.jamesfmackenzie.com/chocolatekeen}                                   & Video game   \\
factorial      & \url{https://www.hellorust.com/demos/factorial/index.html}                            & Mathematics  \\
ffmpeg         & \url{https://w3reality.github.io/async-thread-worker/examples/wasm-ffmpeg/index.html} & Media        \\
figma          & \url{https://www.figma.com}                                                           & Graphics     \\
filament       & \url{https://google.github.io/filament/webgl/demo_suzanne.html}                       & Graphics     \\
funkykarts     & \url{https://www.funkykarts.rocks/demo.html}                                          & Video game   \\
hydro          & \url{https://cselab.github.io/aphros/wasm/hydro.html}                                 & Simulator    \\
imageconvolute & \url{https://takahirox.github.io/WebAssembly-benchmark/tests/imageConvolute.html}     & Benchmark    \\
jqkungfu       & \url{http://jqkungfu.com}                                                             & Programming  \\
jsc            & \url{https://mbbill.github.io/JSC.js/demo/index.html}                                 & Programming  \\
mandelbrot     & \url{http://whealy.com/Rust/mandelbrot.html}                                          & Graphics     \\
ogv-opus       & \url{https://brooke.vibber.net/misc/ogv.js/demo}                                      & Media        \\
ogv-vp9        & \url{https://brooke.vibber.net/misc/ogv.js/demo}                                      & Media        \\
onnx           & \url{https://microsoft.github.io/onnxjs-demo/#}                                       & ML           \\
pacalc         & \url{http://whealy.com/acoustics/PA_Calculator/index.html}                            & Mathematics  \\
parquet        & \url{https://google.github.io/filament/webgl/parquet.html}                            & Graphics     \\
rfxgen         & \url{https://raylibtech.itch.io/rfxgen}                                               & Utility      \\
rguiicons      & \url{https://raylibtech.itch.io/rguiicons}                                            & Utility      \\
rguilayout     & \url{https://raylibtech.itch.io/rguilayout}                                           & Utility      \\
rguistyler     & \url{https://raylibtech.itch.io/rguistyler}                                           & Utility      \\
riconpacker    & \url{https://raylibtech.itch.io/riconpacker}                                          & Utility      \\
rtexpacker     & \url{https://raylibtech.itch.io/rtexpacker}                                           & Utility      \\
rtexviewer     & \url{https://raylibtech.itch.io/rtexviewer}                                           & Utility      \\
sandspiel      & \url{https://sandspiel.club}                                                          & Video game   \\
sqlgui         & \url{http://kripken.github.io/sql.js/examples/GUI}                                    & Programming  \\
sqlpractice    & \url{https://www.sql-practice.com}                                                    & Programming  \\
wasm-astar     & \url{https://jacobdeichert.github.io/wasm-astar}                                      & Benchmark \\ \hline
\end{tabular}
\end{adjustbox}
\label{table:nonminersdesc}
\end{table}

\paragraph*{Obfuscation Methods}

To evaluate effectiveness against obfuscation, we found three publicly available obfuscators supporting WebAssembly which were used in previous works on WebAssembly obfuscation \cite{CABRERAARTEAGA2023103296}, \cite{harnes2024cryptic}: \texttt{Tigress} \cite{tigressHome}, \texttt{emcc-obf} \cite{harnes2024cryptic}, and \texttt{wasm-mutuate} \cite{CABRERAARTEAGA2023103296}. We found that wasm-mutate generates buggy binaries that crash at runtime due to being outdated, and many Tigress obfuscations do not work when applied alongside Wasabi instrumentation. We made the best attempt at applying the following obfuscations to as many of our sample miners as possible.

\begin{itemize}
    \item \textbf{Tigress.} \cite{tigressHome} (Version 4.0.10) Tigress is a source-to-source C obfuscator that has been shown to be effective in software protection \cite{10.1007/978-981-33-4706-9_1} and evading cryptominer detection \cite{harnes2024cryptic}. We apply the following tigress obfuscations to three of the miners which successfully ran with Tigress and Wasabi instrumentation applied.
        \begin{itemize}
            \item \textbf{Encode Arithmetic} - Replace integer arithmetic with more complex expressions, encoded with Mixed Boolean Expressions
            \item \textbf{Function Splitting and Flattening} - Splits functions into smaller fragments and removes structured control flow.
        \end{itemize}
    \item \textbf{emcc-obf.} \cite{harnes2024cryptic} emcc-obf is an obfuscator based on Obfuscator-LLVM \cite{ieeespro2015-JunodRWM} and the Hikari project \cite{githubGitHubHikariObfuscatorHikari}. It is implemented as middle-end passes in the LLVM toolchain. It was demonstrated to be effective in preventing reverse engineering \cite{10.1145/3129676.3129708} and cryptominer detection \cite{harnes2024cryptic}. We chose the following four obfuscations which have been shown in \cite{harnes2024cryptic} to be effective against cryptominer detection.
        \begin{itemize}
            \item \textbf{Bogus Control flow} - Inserting spurious basic blocks and conditional jumps with opaque predicates.
            \item \textbf{Control Flow Flattening} - Removes structured control flow. Similar to Tigress flattening obfuscation.
            \item \textbf{Basic Block Splitting} - Splits LLVM basic blocks into multiple blocks.
            \item \textbf{Substitute Instructions} - Replaces arithmetic and boolean expressions with equivalent but more complex expressions, similarly to Tigress encode arithmetic obfuscation.
        \end{itemize}
\end{itemize}

Some of these obfuscation significantly changes the structure of the data-flow graph as shown in \autoref{figure:cn-all}. Therefore, it is challenging to compare the graphs, and traditional metrics such as graph edit distance are ineffective. Our results show that the \textit{n-fragment inclusion score} is able to effectively compare these graphs.

\paragraph*{Non-miners}

To test the ability to differentiate between mining and non-mining behavior, we evaluated our method against a sample of real-world WebAssembly web applications presented by Baek et al. \cite{DBLP:journals/pacmpl/BaekGS0TRP24}. This data set contains a wide variety of web applications from the \textit{Made with WebAssembly} list \cite{madewithwebassemblyMadeWith}. We exclude samples which do not run for various reasons, e.g. Wasm binaries that contain unsupported extensions for Wasabi, runtime errors with Wasabi, web apps that do not generate a graph, and dead links. In total, we tested 29 real-world Wasm web applications.

To clarify the differences between our benchmarks and Wasm-R3, we included six more programs (filament, imageconvolute, ogv-opus, ogv-vp9, onnx, sqlpractice) and excluded four programs (fib, game-of-life, multiplydouble, multiplyint). Note that ``pathfinding'' and ``guiicons'' in wasm-r3 are named ``wasm-astar'' and ``rguiicons'' in our table. To account for these differences, we tested all the evaluation targets presented in the Wasm-R3 paper, including those that failed their experiments. The programs that failed their evaluation but functioned properly for our Wasabi analysis are imageconvolute, onnx, sqlpractice, and ogv (which contains ogv-opus and ogv-vp9 for audio and video decoding demos). We added ``filament'' to our experiments due to the lack of variety of graphics benchmarks in Wasm-R3 and the importance of graphics computations. Filament is a popular open-source graphics engine developed by Google and has over 18,000 stars on GitHub. Finally, we excluded the four programs because they produced empty data flow graphs with our Wasabi analysis, and it is difficult to ensure whether this is due to a Wasabi bug or the nature of the program itself. They would be trivially classified as non-miners regardless. We note that we used a fork of the Wasabi framework by the Wasm-R3 authors since it is already an improvement over upstream in terms of compatibility. The details of our samples are presented in \autoref{table:nonminersdesc}. For each website, we manually instrument Wasm binaries and inject them back into the website using Google Chrome developer tools.

\begin{table}[t]
    \centering
    \caption{A table of graph sizes for the all benchmark samples. The number of vertices and edges before and after reduction are given, as well as the percent reduction.}
    \tiny
    \begin{tabular}{l|cccccc}
    \hline
        sample & |V| & |E| & |V'| & |E'| & -|V|\% & -|E|\% \\ \hline
        btc & 119 & 202 & 18 & 27 & -84.9\% & -86.6\% \\ 
        zny & 1034 & 2002 & 14 & 22 & -98.6\% & -98.9\% \\ 
        cn & 1326 & 2002 & 8 & 9 & -99.4\% & -99.6\% \\ 
        eth & 1207 & 2002 & 114 & 276 & -90.6\% & -86.2\% \\ 
        wmp & 1045 & 2002 & 11 & 15 & -98.9\% & -99.3\% \\ 
        xmr & 1324 & 2002 & 5 & 5 & -99.6\% & -99.8\% \\ 
        btc-emccobf-boguscf & 1037 & 2002 & 15 & 21 & -98.6\% & -99.0\% \\ 
        btc-emccobf-flatten & 1033 & 2002 & 15 & 22 & -98.5\% & -98.9\% \\ 
        btc-emccobf-split & 1031 & 2002 & 14 & 20 & -98.6\% & -99.0\% \\ 
        btc-emccobf-substitute & 1065 & 2002 & 39 & 80 & -96.3\% & -96.0\% \\ 
        btc-tigress-encodearith & 1063 & 2002 & 21 & 39 & -98.0\% & -98.1\% \\ 
        btc-tigress-splitflatten & 1018 & 2002 & 13 & 18 & -98.7\% & -99.1\% \\ 
        zny-emccobf-boguscf & 1111 & 2002 & 39 & 77 & -96.5\% & -96.2\% \\ 
        zny-emccobf-flatten & 1022 & 2002 & 65 & 139 & -93.6\% & -93.1\% \\ 
        zny-emccobf-split & 1018 & 2002 & 12 & 17 & -98.8\% & -99.2\% \\ 
        zny-emccobf-substitute & 1067 & 2002 & 22 & 56 & -97.9\% & -97.2\% \\ 
        zny-tigress-encodearith & 1066 & 2002 & 16 & 42 & -98.5\% & -97.9\% \\ 
        zny-tigress-splitflatten & 1021 & 2002 & 71 & 153 & -93.0\% & -92.4\% \\ 
        cn-emccobf-boguscf & 1336 & 2002 & 7 & 7 & -99.5\% & -99.7\% \\ 
        cn-emccobf-flatten & 1325 & 2002 & 6 & 6 & -99.5\% & -99.7\% \\ 
        cn-emccobf-split & 1326 & 2002 & 8 & 9 & -99.4\% & -99.6\% \\ 
        cn-emccobf-substitute & 2191 & 2002 & 98 & 264 & -95.5\% & -86.8\% \\ 
        cn-tigress-encodearith & 1252 & 2002 & 9 & 12 & -99.3\% & -99.4\% \\ 
        cn-tigress-splitflatten & 1283 & 2002 & 7 & 10 & -99.5\% & -99.5\% \\ 
        eth-emccobf-boguscf & 1212 & 2002 & 117 & 295 & -90.3\% & -85.3\% \\ 
        eth-emccobf-flatten & 1217 & 2002 & 119 & 316 & -90.2\% & -84.2\% \\ 
        eth-emccobf-split & 1212 & 2002 & 101 & 242 & -91.7\% & -87.9\% \\ 
        eth-emccobf-substitute & 1505 & 2002 & 100 & 417 & -93.4\% & -79.2\% \\ 
        wmp-emccobf-boguscf & 1049 & 2002 & 11 & 12 & -99.0\% & -99.4\% \\ 
        wmp-emccobf-flatten & 1046 & 2002 & 11 & 15 & -98.9\% & -99.3\% \\ 
        wmp-emccobf-split & 1039 & 2002 & 6 & 5 & -99.4\% & -99.8\% \\ 
        wmp-emccobf-substitute & 1044 & 2002 & 12 & 16 & -98.9\% & -99.2\% \\ 
        xmr-emccobf-boguscf & 1296 & 2002 & 16 & 30 & -98.8\% & -98.5\% \\ 
        xmr-emccobf-flatten & 1325 & 2002 & 7 & 8 & -99.5\% & -99.6\% \\ 
        xmr-emccobf-split & 1324 & 2002 & 5 & 5 & -99.6\% & -99.8\% \\ 
        xmr-emccobf-substitute & 2194 & 2002 & 101 & 301 & -95.4\% & -85.0\% \\ 
        boa & 1032 & 2002 & 13 & 16 & -98.7\% & -99.2\% \\ 
        bullet & 179 & 202 & 16 & 25 & -91.1\% & -87.6\% \\ 
        chocolatekeen & 104 & 202 & 7 & 6 & -93.3\% & -97.0\% \\ 
        factorial & 1030 & 2002 & 17 & 27 & -98.3\% & -98.7\% \\ 
        ffmpeg & 1110 & 1972 & 36 & 87 & -96.8\% & -95.6\% \\ 
        figma & 652 & 1270 & 12 & 19 & -98.2\% & -98.5\% \\ 
        filament & 1046 & 2002 & 9 & 15 & -99.1\% & -99.3\% \\ 
        funkykarts & 1006 & 2002 & 7 & 5 & -99.3\% & -99.8\% \\ 
        hydro & 1072 & 1878 & 27 & 63 & -97.5\% & -96.6\% \\ 
        imageconvolute & 1005 & 2002 & 5 & 4 & -99.5\% & -99.8\% \\ 
        jqkungfu & 1005 & 2002 & 5 & 4 & -99.5\% & -99.8\% \\ 
        jsc & 1019 & 2002 & 19 & 29 & -98.1\% & -98.6\% \\ 
        mandelbrot & 1003 & 2002 & 3 & 2 & -99.7\% & -99.9\% \\ 
        ogv-opus & 1126 & 1998 & 21 & 40 & -98.1\% & -98.0\% \\ 
        ogv-vp9 & 1139 & 1888 & 11 & 16 & -99.0\% & -99.2\% \\ 
        onnx & 1008 & 2002 & 3 & 2 & -99.7\% & -99.9\% \\ 
        pacalc & 1090 & 2002 & 130 & 311 & -88.1\% & -84.5\% \\ 
        parquet & 1046 & 2002 & 9 & 15 & -99.1\% & -99.3\% \\ 
        rfxgen & 1032 & 2002 & 11 & 14 & -98.9\% & -99.3\% \\ 
        rguiicons & 1057 & 2002 & 14 & 20 & -98.7\% & -99.0\% \\ 
        rguilayout & 1041 & 2002 & 8 & 9 & -99.2\% & -99.6\% \\ 
        rguistyler & 1041 & 2002 & 7 & 7 & -99.3\% & -99.7\% \\ 
        riconpacker & 1030 & 2002 & 22 & 52 & -97.9\% & -97.4\% \\ 
        rtexpacker & 1062 & 2002 & 11 & 16 & -99.0\% & -99.2\% \\ 
        rtexviewer & 1061 & 2000 & 9 & 12 & -99.2\% & -99.4\% \\ 
        sandspiel & 1014 & 2002 & 10 & 13 & -99.0\% & -99.4\% \\ 
        sqlgui & 1205 & 1840 & 36 & 69 & -97.0\% & -96.2\% \\ 
        sqlpractice & 884 & 1712 & 8 & 8 & -99.1\% & -99.5\% \\ 
        wasm-astar & 1095 & 2002 & 34 & 67 & -96.9\% & -96.7\% \\ 
        \hline
    \end{tabular}
    \label{table:reduction}
\end{table}

\paragraph*{Baseline Comparisons}

We compare our work with three previous publications which represents the state-of-the-art in cryptominer detection: MINOS \cite{naseem2021minos}, Minesweeper \cite{10.1145/3243734.3243858}, and WASim \cite{romano2020wasim}. Among recent works on cryptominer detection, these were the few with working and publicly available artifacts.

\noindent\\
\textbf{MINOS.} MINOS detects cryptomining using a convolutional neural network classifier on gray-scale image representations of WebAssembly binaries \cite{naseem2021minos}. The authors of MINOS reported high detection and low false positive rates, though it has been shown that binary diversification can effectively evade MINOS \cite{CABRERAARTEAGA2023103296}. We use a re-implementation of MINOS by Cabrera-Arteaga et al. \cite{CABRERAARTEAGA2023103296} which has been re-trained in Cryptic Bytes \cite{harnes2024cryptic} using a larger dataset.

\noindent\\
\textbf{Minesweeper.} Minesweeper detects cryptominer binaries using a set of heuristics based on static and intrinsic features of cryptomining code \cite{10.1145/3243734.3243858}. The authors created a set of fingerprints for cryptographic functions by counting cryptography-related instructions and operations. WASM binaries are classified as either being a general cryptominer based on the overall instruction count, or as a CryptoNight miner based on the existence of specific cryptographic primitives, including the following cryptographic functions: Keccak (Keccak 1600-516 and Keccak-f 1600), AES, BLAKE-256, Groestl-256, and Skein-256. These cryptographic primitives are also identified by the distribution of their instruction types.

\noindent\\
\textbf{WASim.} \cite{romano2020wasim} WASim proposed a classification method for WebAssembly binaries using features extracted from Wasm binary files such as function sizes, export types, file attributes, and other metadata. WASim provides four different machine learning models to classify WASM binaries into 11 different categories, including cryptominers. These classifier models include neural network, support vector machine, random forest, and naive Bayes models.

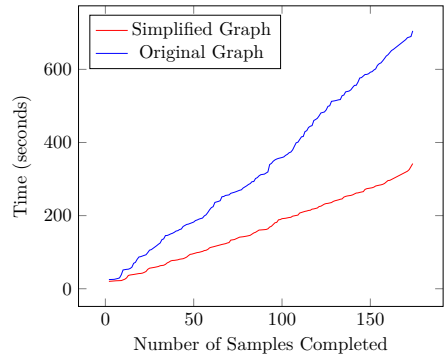
\begin{figure}[t]
    \centering
    \begin{tikzpicture}[scale=0.70]
        \begin{axis}
        [
            xlabel={Number of Samples Completed}, 
            ylabel={Time (seconds)}, 
            legend pos=north west
        ]
        \addplot[red, no marks] table[col sep=comma]{timeseries/nonminers-timeseries.csv};
        \addlegendentry{Simplified Graph}
        \addplot[blue, no marks] table[col sep=comma]{timeseries/unsimplified-nonminers-timeseries.csv};
        \addlegendentry{Original Graph}
        \end{axis}
    \end{tikzpicture}
    \caption{\textbf{(Ablation Study of Analysis Time with and without Simplification)} Comparison of analysis times with and without graph simplification is presented. For the original graphs, this is the time to compute \autoref{algo:approxnfis} (n-FIS score). For simplified graphs, the time includes both \autoref{algo:simplifyapprox} and \autoref{algo:approxnfis}.}
    \label{fig:timeabalation}
\end{figure}

\subsection{RQ1: The Effectiveness of Graph Simplification}

To demonstrate the effectiveness of our graph simplification algorithm, \autoref{table:reduction} shows the reduction in the number of vertices and edges in the graphs of our benchmark samples. Across all 65 graphs in the benchmark, we achieve an average of 97.3\% reduction in vertex count and 96.5\% reduction in edge count. The ablation study in \autoref{fig:timeabalation} shows a significant reduction in the computation time for the n-FIS score even though ArcMatch has been shown to be scalable with respect to the number of vertices and edges \cite{DBLP:journals/datamine/BonniciGMMSPG24}.

\autoref{table:minervsminer}a and \autoref{table:minervsminer}b show the n-FIS scores between each pair of cryptominers in the sample with and without graph simplification, respectively. The scores change only slightly for most entries, except for \verb|eth|, which has a much more complex and less reducible graph than the other samples. The larger simplified graph of \verb|eth|, which contains more than 10 times the number of edges and vertices of other samples, would make it much more prone to contain other fingerprints as a subgraph due to its size and complexity. Although the algorithm preserves structural information and connectivity, combining repeated subgraphs means that we lose most information about the frequency of the substructures, yet the scores demonstrate that this structural information is enough to differentiate between different species of miners in our sample. The detection results in \autoref{table:summary} show that the complete framework of {\ournameshort} achieves better accuracy and $f_1$ score compared to {\ournameshort} without graph simplification, demonstrating that this minimal structural information is not only sufficient, but also amplifies the uniqueness of cryptominer data-flow graphs.

\begin{table}[t]
    \centering
    \caption{\textbf{(Simplified and Unsimplified Miners Pairwise n-FIS Scores)} Pairwise comparison of 5-FIS score between the miner samples with and without graph simplification. The miner in each row is checked whether it contains the miner in each column as subgraphs. Scores above 0.5 and 0.65 are highlight in light red and dark red respectively.}
    \scriptsize
    (a) \textbf{Simplified} Miners Pairwise n-FIS Scores\\
    \begin{tabular}{l|ccccccc}
    \hline
        \textbf{miner} & \textbf{btc} & \textbf{zny} & \textbf{cn} & \textbf{eth} & \textbf{wmp} & \textbf{xmr} \\ \hline
        btc & \cellcolor{red!55}0.915 & 0.249 & 0.091 & 0.119 & \cellcolor{red!15}0.57 & 0.209 \\ 
        zny & 0.121 & \cellcolor{red!55}0.929 & 0.163 & 0.057 & 0.275 & 0.207 \\ 
        cn & 0.023 & 0.154 & \cellcolor{red!55}0.817 & 0.037 & 0.035 & \cellcolor{red!55}0.94 \\ 
        eth & \cellcolor{red!55}0.823 & 0.277 & 0.074 & \cellcolor{red!55}0.972 & \cellcolor{red!55}0.916 & 0.171 \\ 
        wmp & 0.233 & 0.098 & 0.03 & 0.027 & \cellcolor{red!55}0.94 & 0.188 \\ 
        xmr & 0.031 & 0.097 & \cellcolor{red!55}0.709 & 0.029 & 0.01 & \cellcolor{red!55}1.0 \\ \hline
    \end{tabular}
    ~\\~\\~\\
    (b) \textbf{Unsimplified} Miners Pairwise n-FIS Scores\\
    \begin{tabular}{l|ccccccc}
    \hline
        \textbf{miner} & \textbf{btc} & \textbf{zny} & \textbf{cn} & \textbf{eth} & \textbf{wmp} & \textbf{xmr} \\ \hline
        btc & \cellcolor{red!55}0.851 & 0.236 & 0.151 & 0.221 & 0.484 & 0.184 \\ 
        zny & 0.251 & \cellcolor{red!55}0.867 & 0.183 & 0.173 & 0.453 & 0.206 \\ 
        cn & 0.13 & 0.36 & \cellcolor{red!55}0.756 & 0.145 & 0.022 & \cellcolor{red!55}0.743 \\ 
        eth & \cellcolor{red!15}0.551 & 0.265 & 0.143 & \cellcolor{red!55}0.943 & \cellcolor{red!15}0.577 & 0.145 \\ 
        wmp & 0.138 & 0.201 & 0.162 & 0.03 & \cellcolor{red!55}0.83 & 0.225 \\ 
        xmr & 0.118 & 0.303 & \cellcolor{red!55}0.855 & 0.161 & 0.011 & \cellcolor{red!55}0.821 \\ \hline
    \end{tabular}
    \label{table:minervsminer}
\end{table}

\subsection{RQ2: The Effectiveness of {\ournameshort}}

To demonstrate the effectiveness of {\ournameshort} and the \textit{n-fragment inclusion score}, we test its ability to distinguish cryptomining behavior in the presence of obfuscation, as well as differentiate between miners and non-miners. \autoref{table:minervsminer}a shows the n-FIS score of the miner samples against themselves. The scores reflect clear differences between the different species of miners. We observe that \verb|cn| and \verb|xmr| both perform the same variant of the \verb|CryptoNight| algorithm and therefore share mutually high inclusion score. Note that the inclusion scores are not necessarily symmetric, as seen with \verb|wmp| and \verb|eth|. Although both algorithms use the \verb|keccak| family of hashing algorithms as a component, they perform additional work to achieve memory-hardness \cite{getmoneroCryptoNightMonero, ethashkernel}. We hypothesize that the additional work performed by \verb|wmp| distinct from \verb|eth| does not include the three instrumented instructions \verb|and|, \verb|xor|, and \verb|shr|, giving the impression that the work performed by \verb|wmp| is a subset of \verb|eth|.

The first six rows of \autoref{table:baselineobf} shows the performance of our baselines in the six miner samples. MINOS and Minesweeper are able to recognize all six samples as cryptominers, while the four models provided by WASim struggle to identify our samples.

\begin{table}[ht]
    \centering
    \caption{\textbf{(Baseline vs.~All Benchmark Samples)} Baseline detection results for miners, obfuscated miners, and non-miners. 1 and 0 represent malign and benign results respectively. The columns show results for MINOS, Minesweeper, WASim neural network, WASim random forest, WASim support vector, and WASim naive bayes classifiers from left to right.}
    \scriptsize
    \begin{tabular}{|l|l|l|c|c|c|c|c|c|}
    \hline
        \textbf{miner} & \textbf{obf} & \textbf{strat} & \textbf{minos} & \textbf{mswp} & \textbf{wsnn} & \textbf{wsrf} & \textbf{wssv} & \textbf{wsnb} \\ \hline
        zny & - & - & \cellcolor{red!55}1 & \cellcolor{red!55}1 & 0 & 0 & 0 & \cellcolor{red!55}1 \\ 
        cn  & - & - & \cellcolor{red!55}1 & \cellcolor{red!55}1 & \cellcolor{red!55}1 & \cellcolor{red!55}1 & 0 & 0 \\ 
        btc & - & - & \cellcolor{red!55}1 & \cellcolor{red!55}1 & 0 & 0 & 0 & \cellcolor{red!55}1 \\ 
        eth & - & - & \cellcolor{red!55}1 & \cellcolor{red!55}1 & 0 & 0 & 0 & \cellcolor{red!55}1 \\ 
        wmp & - & - & \cellcolor{red!55}1 & \cellcolor{red!55}1 & 0 & \cellcolor{red!55}1 & 0 & 0 \\ 
        xmr & - & - & \cellcolor{red!55}1 & \cellcolor{red!55}1 & \cellcolor{red!55}1 & 0 & 0 & 0 \\ 
        btc & emccobf & boguscf & 0 & \cellcolor{red!55}1 & 0 & 0 & 0 & \cellcolor{red!55}1 \\ 
        btc & emccobf & flatten & 0 & \cellcolor{red!55}1 & 0 & 0 & 0 & \cellcolor{red!55}1 \\ 
        btc & emccobf & split & 0 & \cellcolor{red!55}1 & 0 & 0 & 0 & \cellcolor{red!55}1 \\ 
        btc & emccobf & substitute & 0 & \cellcolor{red!55}1 & 0 & 0 & 0 & \cellcolor{red!55}1 \\ 
        btc & tigress & encodearith & 0 & \cellcolor{red!55}1 & 0 & 0 & 0 & \cellcolor{red!55}1 \\ 
        btc & tigress & splitflatten & \cellcolor{red!55}1 & \cellcolor{red!55}1 & 0 & 0 & 0 & \cellcolor{red!55}1 \\ 
        zny & emccobf & boguscf & \cellcolor{red!55}1 & \cellcolor{red!55}1 & 0 & 0 & 0 & \cellcolor{red!55}1 \\ 
        zny & emccobf & flatten & 0 & \cellcolor{red!55}1 & 0 & 0 & 0 & \cellcolor{red!55}1 \\ 
        zny & emccobf & split & 0 & \cellcolor{red!55}1 & 0 & 0 & 0 & \cellcolor{red!55}1 \\ 
        zny & emccobf & substitute & \cellcolor{red!55}1 & \cellcolor{red!55}1 & 0 & 0 & 0 & \cellcolor{red!55}1 \\ 
        zny & tigress & encodearith & \cellcolor{red!55}1 & \cellcolor{red!55}1 & 0 & 0 & 0 & 0 \\ 
        zny & tigress & splitflatten & 0 & \cellcolor{red!55}1 & 0 & 0 & 0 & \cellcolor{red!55}1 \\ 
        cn & tigress & encodearith & \cellcolor{red!55}1 & \cellcolor{red!55}1 & 0 & \cellcolor{red!55}1 & 0 & 0 \\ 
        cn & tigress & splitflatten & 0 & 0 & 0 & 0 & 0 & 0 \\ 
        cn & emccobf & boguscf & \cellcolor{red!55}1 & \cellcolor{red!55}1 & 0 & \cellcolor{red!55}1 & 0 & 0 \\ 
        cn & emccobf & flatten & \cellcolor{red!55}1 & 0 & 0 & \cellcolor{red!55}1 & 0 & 0 \\ 
        cn & emccobf & split & \cellcolor{red!55}1 & \cellcolor{red!55}1 & \cellcolor{red!55}1 & \cellcolor{red!55}1 & 0 & 0 \\ 
        cn & emccobf & substitute & 0 & \cellcolor{red!55}1 & \cellcolor{red!55}1 & \cellcolor{red!55}1 & 0 & 0 \\ 
        eth & emccobf & boguscf & 0 & \cellcolor{red!55}1 & 0 & 0 & 0 & \cellcolor{red!55}1 \\ 
        eth & emccobf & flatten & \cellcolor{red!55}1 & \cellcolor{red!55}1 & 0 & 0 & 0 & \cellcolor{red!55}1 \\ 
        eth & emccobf & split & 0 & \cellcolor{red!55}1 & 0 & 0 & 0 & \cellcolor{red!55}1 \\ 
        eth & emccobf & substitute & \cellcolor{red!55}1 & \cellcolor{red!55}1 & 0 & 0 & 0 & \cellcolor{red!55}1 \\ 
        wmp & emccobf & boguscf & \cellcolor{red!55}1 & \cellcolor{red!55}1 & 0 & \cellcolor{red!55}1 & 0 & 0 \\ 
        wmp & emccobf & flatten & \cellcolor{red!55}1 & \cellcolor{red!55}1 & 0 & \cellcolor{red!55}1 & 0 & \cellcolor{red!55}1 \\ 
        wmp & emccobf & split & \cellcolor{red!55}1 & \cellcolor{red!55}1 & 0 & \cellcolor{red!55}1 & 0 & 0 \\ 
        wmp & emccobf & substitute & \cellcolor{red!55}1 & \cellcolor{red!55}1 & 0 & \cellcolor{red!55}1 & 0 & \cellcolor{red!55}1 \\ 
        xmr & emccobf & boguscf & \cellcolor{red!55}1 & \cellcolor{red!55}1 & 0 & \cellcolor{red!55}1 & 0 & 0 \\ 
        xmr & emccobf & flatten & 0 & \cellcolor{red!55}1 & 0 & 0 & 0 & 0 \\ 
        xmr & emccobf & split & \cellcolor{red!55}1 & \cellcolor{red!55}1 & \cellcolor{red!55}1 & 0 & 0 & 0 \\ 
        xmr & emccobf & substitute & 0 & \cellcolor{red!55}1 & \cellcolor{red!55}1 & 0 & 0 & 0 \\ \hline \multicolumn{1}{l}{} \\ \hline

        \textbf{nonminer} & \textbf{obf} & \textbf{strat} & \textbf{minos} & \textbf{mswp} & \textbf{wsnn} & \textbf{wsrf} & \textbf{wssv} & \textbf{wsnb} \\ \hline
        boa & - & - & 0 & \cellcolor{red!55}1 & 0 & 0 & 0 & \cellcolor{red!55}1 \\ 
        bullet & - & - & 0 & \cellcolor{red!55}1 & 0 & 0 & 0 & \cellcolor{red!55}1 \\ 
        chocolatekeen & - & - & \cellcolor{red!55}1 & \cellcolor{red!55}1 & 0 & 0 & 0 & 0 \\ 
        factorial & - & - & 0 & 0 & 0 & 0 & 0 & 0 \\ 
        ffmpeg & - & - & 0 & \cellcolor{red!55}1 & 0 & 0 & 0 & \cellcolor{red!55}1 \\ 
        figma & - & - & 0 & \cellcolor{red!55}1 & 0 & 0 & 0 & 0 \\ 
        filament & - & - & 0 & \cellcolor{red!55}1 & 0 & 0 & 0 & \cellcolor{red!55}1 \\ 
        funkykarts & - & - & 0 & \cellcolor{red!55}1 & 0 & 0 & 0 & \cellcolor{red!55}1 \\ 
        hydro & - & - & 0 & \cellcolor{red!55}1 & 0 & 0 & 0 & \cellcolor{red!55}1 \\ 
        imageconvolute & - & - & \cellcolor{red!55}1 & \cellcolor{red!55}1 & 0 & 0 & 0 & 0 \\ 
        jqkungfu & - & - & 0 & \cellcolor{red!55}1 & 0 & 0 & 0 & 0 \\ 
        jsc & - & - & 0 & \cellcolor{red!55}1 & 0 & 0 & 0 & \cellcolor{red!55}1 \\ 
        mandelbrot & - & - & \cellcolor{red!55}1 & \cellcolor{red!55}1 & 0 & \cellcolor{red!55}1 & 0 & 0 \\ 
        ogv-opus & - & - & 0 & \cellcolor{red!55}1 & 0 & \cellcolor{red!55}1 & 0 & \cellcolor{red!55}1 \\ 
        ogv-vp9 & - & - & 0 & \cellcolor{red!55}1 & 0 & 0 & 0 & 0 \\ 
        onnx & - & - & 0 & \cellcolor{red!55}1 & \cellcolor{red!55}1 & \cellcolor{red!55}1 & 0 & 0 \\ 
        pacalc & - & - & \cellcolor{red!55}1 & \cellcolor{red!55}1 & 0 & 0 & 0 & 0 \\ 
        parquet & - & - & 0 & \cellcolor{red!55}1 & 0 & 0 & 0 & \cellcolor{red!55}1 \\ 
        rfxgen & - & - & 0 & \cellcolor{red!55}1 & 0 & 0 & 0 & \cellcolor{red!55}1 \\ 
        rguiicons & - & - & 0 & \cellcolor{red!55}1 & 0 & 0 & 0 & \cellcolor{red!55}1 \\ 
        rguilayout & - & - & 0 & \cellcolor{red!55}1 & 0 & 0 & 0 & \cellcolor{red!55}1 \\ 
        rguistyler & - & - & 0 & \cellcolor{red!55}1 & 0 & 0 & 0 & \cellcolor{red!55}1 \\ 
        riconpacker & - & - & 0 & \cellcolor{red!55}1 & 0 & 0 & 0 & \cellcolor{red!55}1 \\ 
        rtexpacker & - & - & 0 & \cellcolor{red!55}1 & 0 & 0 & 0 & \cellcolor{red!55}1 \\ 
        rtexviewer & - & - & 0 & \cellcolor{red!55}1 & 0 & 0 & 0 & \cellcolor{red!55}1 \\ 
        sandspiel & - & - & \cellcolor{red!55}1 & \cellcolor{red!55}1 & 0 & 0 & 0 & \cellcolor{red!55}1 \\ 
        sqlgui & - & - & 0 & \cellcolor{red!55}1 & 0 & 0 & 0 & \cellcolor{red!55}1 \\ 
        sqlpractice & - & - & 0 & \cellcolor{red!55}1 & 0 & 0 & 0 & \cellcolor{red!55}1 \\ 
        wasm-astar & - & - & 0 & \cellcolor{red!55}1 & 0 & 0 & 0 & 0 \\ \hline
    \end{tabular}
    \label{table:baselineobf}
\end{table}

\begin{table}[t]
    \centering
    \caption{\textbf{(Benchmark Samples n-FIS scores)} The 5-FIS scores showing the inclusion of each original miner fingerprint (columns) inside each benchmark sample (rows) is shown. The last column shows the detection result for a 0.65 5-FIS score threshold in any column. Cells above 0.5 and 0.65 scores are highlighted in light red and dark red respectively.}
    \scriptsize
    \begin{tabular}{|l|l|l|c|c|c|c|c|c|c|}
    \hline
        \textbf{miner} & \textbf{obf} & \textbf{strat} & \textbf{btc} & \textbf{zny} & \textbf{cn} & \textbf{eth} & \textbf{wmp} & \textbf{xmr} & $\ge$\textbf{0.65} \\ \hline
        btc & emccobf & boguscf & \cellcolor{red!55}0.903 & 0.267 & 0.081 & 0.072 & \cellcolor{red!15}0.57 & 0.177 & \cellcolor{red!55}1 \\ 
        btc & emccobf & flatten & \cellcolor{red!55}0.695 & 0.301 & 0.092 & 0.082 & \cellcolor{red!15}0.55 & 0.187 & \cellcolor{red!55}1 \\ 
        btc & emccobf & split & \cellcolor{red!55}0.653 & 0.256 & 0.081 & 0.073 & \cellcolor{red!15}0.534 & 0.193 & \cellcolor{red!55}1 \\ 
        btc & emccobf & substitute & \cellcolor{red!55}0.846 & 0.326 & 0.475 & 0.217 & \cellcolor{red!55}0.773 & 0.453 & \cellcolor{red!55}1 \\ 
        btc & tigress & encodearith & \cellcolor{red!55}0.777 & 0.451 & \cellcolor{red!15}0.527 & 0.162 & \cellcolor{red!55}0.733 & \cellcolor{red!15}0.601 & \cellcolor{red!55}1 \\ 
        btc & tigress & splitflatten & \cellcolor{red!55}0.733 & 0.275 & 0.093 & 0.058 & \cellcolor{red!55}0.671 & 0.193 & \cellcolor{red!55}1 \\ 
        zny & emccobf & boguscf & 0.197 & \cellcolor{red!55}0.821 & 0.187 & \cellcolor{red!15}0.613 & \cellcolor{red!15}0.519 & 0.161 & \cellcolor{red!55}1 \\ 
        zny & emccobf & flatten & 0.224 & \cellcolor{red!55}0.783 & 0.211 & \cellcolor{red!15}0.619 & 0.499 & 0.154 & \cellcolor{red!55}1 \\ 
        zny & emccobf & split & 0.145 & \cellcolor{red!55}0.821 & 0.157 & 0.069 & \cellcolor{red!15}0.545 & 0.064 & \cellcolor{red!55}1 \\ 
        zny & emccobf & substitute & 0.38 & \cellcolor{red!55}0.911 & \cellcolor{red!55}0.674 & 0.269 & \cellcolor{red!55}0.896 & \cellcolor{red!55}0.759 & \cellcolor{red!55}1 \\ 
        zny & tigress & encodearith & 0.38 & 0.307 & \cellcolor{red!15}0.511 & 0.059 & \cellcolor{red!55}0.804 & \cellcolor{red!15}0.598 & \cellcolor{red!55}1 \\ 
        zny & tigress & splitflatten & 0.266 & \cellcolor{red!55}0.739 & 0.189 & \cellcolor{red!15}0.631 & \cellcolor{red!55}0.705 & 0.217 & \cellcolor{red!55}1 \\ 
        cn & emccobf & boguscf & 0.03 & 0.106 & \cellcolor{red!15}0.647 & 0.04 & 0.04 & \cellcolor{red!55}0.951 & \cellcolor{red!55}1 \\ 
        cn & emccobf & flatten & 0.031 & 0.092 & \cellcolor{red!55}0.748 & 0.044 & 0.024 & \cellcolor{red!55}1.0 & \cellcolor{red!55}1 \\ 
        cn & emccobf & split & 0.038 & 0.151 & \cellcolor{red!55}0.836 & 0.039 & 0.029 & \cellcolor{red!55}0.945 & \cellcolor{red!55}1 \\ 
        cn & emccobf & substitute & \cellcolor{red!55}0.945 & \cellcolor{red!55}0.777 & \cellcolor{red!55}0.824 & \cellcolor{red!55}0.896 & \cellcolor{red!55}0.877 & \cellcolor{red!55}0.88 & \cellcolor{red!55}1 \\ 
        cn & tigress & encodearith & 0.029 & 0.197 & \cellcolor{red!55}0.792 & 0.063 & 0.087 & \cellcolor{red!55}0.735 & \cellcolor{red!55}1 \\ 
        cn & tigress & splitflatten & 0.027 & 0.181 & \cellcolor{red!55}0.966 & 0.047 & 0.146 & \cellcolor{red!55}0.985 & \cellcolor{red!55}1 \\ 
        eth & emccobf & boguscf & \cellcolor{red!55}0.81 & 0.271 & 0.129 & \cellcolor{red!55}0.97 & \cellcolor{red!55}0.912 & 0.165 & \cellcolor{red!55}1 \\ 
        eth & emccobf & flatten & \cellcolor{red!55}0.801 & 0.358 & 0.057 & \cellcolor{red!55}0.967 & \cellcolor{red!55}0.886 & 0.099 & \cellcolor{red!55}1 \\ 
        eth & emccobf & split & \cellcolor{red!55}0.801 & 0.339 & 0.068 & \cellcolor{red!55}0.967 & \cellcolor{red!55}0.909 & 0.181 & \cellcolor{red!55}1 \\ 
        eth & emccobf & substitute & \cellcolor{red!55}0.957 & \cellcolor{red!55}0.919 & 0.203 & \cellcolor{red!55}0.971 & \cellcolor{red!55}0.879 & 0.209 & \cellcolor{red!55}1 \\ 
        wmp & emccobf & boguscf & 0.195 & 0.069 & 0.015 & 0.023 & \cellcolor{red!55}0.858 & 0.109 & \cellcolor{red!55}1 \\ 
        wmp & emccobf & flatten & 0.195 & 0.089 & 0.023 & 0.025 & \cellcolor{red!55}0.836 & 0.163 & \cellcolor{red!55}1 \\ 
        wmp & emccobf & split & 0.126 & 0.02 & 0.0 & 0.015 & 0.344 & 0.0 & 0 \\ 
        wmp & emccobf & substitute & 0.233 & 0.106 & 0.029 & 0.027 & \cellcolor{red!55}0.935 & 0.146 & \cellcolor{red!55}1 \\ 
        xmr & emccobf & boguscf & 0.23 & 0.383 & \cellcolor{red!55}0.815 & 0.203 & \cellcolor{red!15}0.615 & \cellcolor{red!55}0.878 & \cellcolor{red!55}1 \\ 
        xmr & emccobf & flatten & 0.032 & 0.129 & \cellcolor{red!55}0.817 & 0.06 & 0.04 & \cellcolor{red!55}0.948 & \cellcolor{red!55}1 \\ 
        xmr & emccobf & split & 0.027 & 0.092 & \cellcolor{red!55}0.714 & 0.042 & 0.04 & \cellcolor{red!55}1.0 & \cellcolor{red!55}1 \\ 
        xmr & emccobf & substitute & \cellcolor{red!55}0.961 & \cellcolor{red!55}0.718 & \cellcolor{red!55}0.81 & \cellcolor{red!55}0.909 & \cellcolor{red!55}0.889 & \cellcolor{red!55}0.876 & \cellcolor{red!55}1 \\ \hline \multicolumn{1}{l}{} \\ \hline

        \textbf{nonminer} & \textbf{obf} & \textbf{strat} & \textbf{btc} & \textbf{zny} & \textbf{cn} & \textbf{eth} & \textbf{wmp} & \textbf{xmr} & $\ge$\textbf{0.65} \\ \hline
        boa & - & - & 0.017 & 0.214 & 0.309 & 0.005 & 0.011 & 0.415 & 0 \\ 
        bullet & - & - & 0.082 & 0.162 & 0.021 & 0.023 & 0.449 & 0.169 & 0 \\ 
        chocolatekeen & - & - & 0.062 & 0.091 & 0.0 & 0.001 & 0.497 & 0.0 & 0 \\ 
        factorial & - & - & 0.34 & 0.447 & 0.378 & 0.034 & \cellcolor{red!15}0.639 & \cellcolor{red!15}0.533 & 0 \\ 
        ffmpeg & - & - & 0.366 & \cellcolor{red!15}0.558 & 0.382 & 0.069 & \cellcolor{red!15}0.573 & \cellcolor{red!15}0.569 & 0 \\ 
        figma & - & - & 0.016 & 0.192 & 0.398 & 0.001 & 0.009 & 0.369 & 0 \\ 
        filament & - & - & 0.033 & 0.143 & \cellcolor{red!15}0.526 & 0.005 & 0.014 & \cellcolor{red!15}0.597 & 0 \\ 
        funkykarts & - & - & 0.014 & 0.031 & 0.017 & 0.002 & 0.015 & 0.061 & 0 \\ 
        hydro & - & - & 0.138 & 0.377 & 0.373 & 0.033 & 0.412 & \cellcolor{red!15}0.539 & 0 \\ 
        imageconvolute & - & - & 0.007 & 0.013 & 0.011 & 0.0 & 0.021 & 0.073 & 0 \\ 
        jqkungfu & - & - & 0.005 & 0.013 & 0.007 & 0.0 & 0.02 & 0.065 & 0 \\ 
        jsc & - & - & 0.058 & 0.244 & 0.239 & 0.006 & 0.347 & 0.291 & 0 \\ 
        mandelbrot & - & - & 0.009 & 0.033 & 0.0 & 0.0 & 0.064 & 0.0 & 0 \\ 
        ogv-opus & - & - & 0.164 & 0.486 & 0.362 & 0.031 & 0.427 & \cellcolor{red!15}0.551 & 0 \\ 
        ogv-vp9 & - & - & 0.1 & 0.37 & 0.494 & 0.018 & 0.119 & \cellcolor{red!15}0.602 & 0 \\ 
        onnx & - & - & 0.005 & 0.014 & 0.008 & 0.0 & 0.013 & 0.065 & 0 \\ 
        pacalc & - & - & 0.048 & 0.271 & 0.385 & 0.014 & 0.019 & \cellcolor{red!15}0.525 & 0 \\ 
        parquet & - & - & 0.033 & 0.081 & 0.492 & 0.003 & 0.02 & \cellcolor{red!15}0.576 & 0 \\ 
        rfxgen & - & - & 0.025 & 0.225 & 0.411 & 0.001 & 0.018 & \cellcolor{red!15}0.604 & 0 \\ 
        rguiicons & - & - & 0.061 & 0.201 & 0.379 & 0.029 & 0.306 & \cellcolor{red!15}0.543 & 0 \\ 
        rguilayout & - & - & 0.001 & 0.009 & 0.017 & 0.0 & 0.008 & 0.095 & 0 \\ 
        rguistyler & - & - & 0.0 & 0.01 & 0.02 & 0.0 & 0.014 & 0.103 & 0 \\ 
        riconpacker & - & - & 0.037 & 0.297 & 0.369 & 0.01 & 0.015 & \cellcolor{red!15}0.553 & 0 \\ 
        rtexpacker & - & - & 0.059 & 0.113 & 0.019 & 0.004 & 0.473 & 0.093 & 0 \\ 
        rtexviewer & - & - & 0.061 & 0.115 & 0.013 & 0.005 & 0.463 & 0.099 & 0 \\ 
        sandspiel & - & - & 0.201 & 0.189 & 0.323 & 0.014 & 0.101 & \cellcolor{red!15}0.568 & 0 \\ 
        sqlgui & - & - & 0.115 & 0.425 & 0.377 & 0.028 & 0.449 & \cellcolor{red!15}0.525 & 0 \\ 
        sqlpractice & - & - & 0.031 & 0.107 & 0.359 & 0.004 & 0.01 & 0.352 & 0 \\ 
        wasm-astar & - & - & 0.319 & 0.433 & 0.371 & 0.115 & 0.467 & \cellcolor{red!15}0.549 & 0 \\ \hline
    \end{tabular}
    \label{table:nfisobf}
\end{table}

\afterpage{\clearpage}

\subsubsection{Effectiveness Against Obfuscation}

The first half of \autoref{table:baselineobf} shows the detection results of the six baselines in the obfuscated miner samples. Among these, Minesweeper achieved the highest sensitivity at $94.4\%$, followed by WASim naive Bayes, and MINOS. Although MINOS was able to detect all of our original samples as cryptominers, it struggles to identify the obfuscated samples, especially under function splitting and control flow flattening obfuscations.

The first half of \autoref{table:nfisobf} shows the n-FIS score of each obfuscated miner against the original miners. In the last column we present the {\ournameshort} detection results at $0.65$ score threshold for any match in the miner database. At this threshold, only $1$ obfuscated sample is missed, with a sensitivity of $96.6\%$. The similarity trends present in \autoref{table:minervsminer}a are also present in the obfuscated miner samples, with some exception. Of interest is the instruction substitution obfuscation performed by emcc-obf, which consistently raises the subgraph similarity scores for all the miner fingerprints, across all obfuscated samples. The graphs of CryptoNight with and without substitution obfuscation, shown in \autoref{figure:cn-ems} and \autoref{figure:cn} respectively, suggests that the variety of substituted instructions creates a complex enough graph to contain the behavior patterns of all other miners.

\subsubsection{False Positives}

Finally, we evaluate the performance of the detection methods against a sample of real-world WASM web applications to test for false positives. The results for the baselines are shown in the second half of \autoref{table:baselineobf}. Although WASim neural network, random forest, and support vector machine classifiers all have high specificity of $100\%$, $96.6\%$, and $89.7\%$, respectively, they also classify the majority of malicious samples as benign. Minesweeper on the other hand classifies almost everything as malicious, getting a $3.4\%$ specificity.

The second half of \autoref{table:nfisobf} shows the n-FIS scores and the results for the $0.65$ detection threshold for non-miners. Although we achieve a specificity of $100\%$ at this threshold, we observe that the similarity scores are high for the \verb|xmr| miner due to the simplicity of its data-flow graph. From this, it would be reasonable to apply different detection thresholds for different malicious fingerprint graphs based on their sizes or another simplicity metric.

\subsubsection{Results}

The summarized evaluation metrics for all the detection methods tested are shown in \autoref{table:summary}. {\ournameshort} was able to achieve an overall accuracy of $98.3\%$ at the $0.65$ detection threshold. The best result among the baselines is MINOS with an accuracy of $70.8\%$ and a $f_1$ score of $69.8\%$. WASim neural network, random forest, and support vector machine classifiers skew toward labeling most samples as non-miners, while Minesweeper and WASim naive Bayes classifier label most samples as miners.

\subsection{Discussion}

The result of this evaluation shows that our graph simplification algorithm in conjunction with the n-FIS score demonstrates the ability to differentiate cryptominers from other common types of Wasm web applications based on their data-flow graphs. Moreover, we demonstrate the ability to detect cryptominers under various obfuscations, outperforming the state-of-the-art MINOS.

\begin{table}[t]
    \centering
    \caption{Summary of performance metrics of all tested detection methods.}
    \scriptsize
    \begin{tabular}{|l|c|c|c|c|c|c|c|}
    \hline
        \textbf{Method} & \textbf{Accuracy} & \textbf{Sensitivity} & \textbf{Specificity} & \textbf{Precision} & \textbf{$f_1$-score} \\ \hline
        {\ournameshort} & \cellcolor{teal!15}98.3\% & \cellcolor{teal!15}96.7\% & \cellcolor{teal!15}100.0\% & \cellcolor{teal!15}100.0\% & \cellcolor{teal!15}98.3\% \\ 
        {\ournameshort} (no simplify) & 93.2\% & 93.3\% & 93.1\% & 93.3\% & 93.3\% \\ 
        MINOS & 70.8\% & 61.1\% & 82.8\% & 81.5\% & 69.8\% \\ 
        Minesweeper & 53.8\% & 94.4\% & 3.4\% & 54.8\% & 69.4\% \\ 
        WASim nn & 52.3\% & 16.7\% & 96.6\% & 85.7\% & 27.9\% \\ 
        WASim rf & 58.5\% & 33.3\% & 89.7\% & 80.0\% & 47.1\% \\ 
        WASim svm & 44.6\% & 0.0\% & \cellcolor{teal!15}100.0\% & N/A & 0.0\% \\ 
        WASim nb & 46.2\% & 55.6\% & 34.5\% & 51.3\% & 53.3\% \\ \hline
    \end{tabular}
    \label{table:summary}
\end{table}

\noindent\\
\textbf{RQ1.} In the first part of our evaluation, we showed that \autoref{algo:simplifyapprox} is able to reduce graphs of over 3000 edges and vertices to around 20 in most cases. This reduction significantly speeds up comparison operations between graphs even when employing highly scalable algorithms. Furthermore, the simplification of repetitive structural information preserves the local properties of computation and amplifies the distinction of cryptomining algorithms as evident in the detection metrics of \autoref{table:summary}.

\noindent\\
\textbf{RQ2.} The second part of our evaluation indicate that {\ournameshort} is highly effective in identifying cryptominers even in the precense of obfuscation through the use of the \textit{n-fragment inclusion score}. Although many obfuscations are able to generate substantially different graphs for the same miners, the subgraph similarity score remains high, showing that our local and fragmented view of subgraph similarity is tolerant to noise introduced by common program transformations. While the \verb|xmr| miner exhibits high inclusion scores among benign applications, this can be attributed to the simplicity of its graph and behavior, and we recommend deploying different thresholds based on the simplicity of a fingerprint in real world scenarios.

\noindent\\
\textbf{Performance.} Although performance was not the main objective of our study, \autoref{fig:timeabalation} shows that our analysis is scalable. There are several areas in which we could reduce the analysis overhead in real-world deployment. (1) The resilience of our detection method against fragmentation allows us to collect data-flow traces at random intervals, thus instrumentation overhead could be reduced significantly. (2) We developed our prototype in Python, which could be optimized by using a different language.

\noindent\\
\textbf{Scalability and Limitations.} The approximate graph simplification through backward random walks is designed to solve the scalability issue of analyzing large data-flow graphs, which we have shown to be effective on real world Wasm applications in the ablation study (\autoref{fig:timeabalation} and \autoref{table:reduction}). Since our simplification relies on repeated structures in the graph resulting from repeated computations, programs that perform many diverse types of computation may not simplify very well, but none of our real-world benchmark programs have this issue. Another limitation of our method arises due to its similarity to traditional signature checking, namely that a database of fingerprints needs to be maintained and updated. The scalability of our detection algorithm depends on the number of unique cryptomining algorithms. Fortunately, it has been shown in literature that cryptomining scripts and algorithms are low in diversity \cite{tekiner2021browser, 9092245, kharraz2019outguard}.

\section{Conclusion}


Given the limited diversity of cryptomining scripts in the wild, obfuscation serves as a natural and appealing solution to avoid detection. It is crucial to develop a cryptomining detection method that is resilient to code obfuscation. In this paper, we propose using instruction-level data-flow graphs as a valuable source of information on a program's computational behavior. We present: (1) a graph simplification algorithm to reduce the computational burden of processing large and granular data-flow graphs while preserving local substructures; and (2) a subgraph similarity measure, the \textit{n-fragment inclusion score}, based on fragment inclusion that is robust against noise and obfuscation. Our experimental results demonstrate that the simplified graph fingerprints retain essential structural information that distinguishes proof-of-work algorithms, and the n-fragment inclusion score effectively quantifies this structural difference. The combined framework {\ournameshort} achieved high accuracy against standard obfuscation, outperforming existing detection methods.


\bibliography{lipics-v2021-sample-article}

\newpage

\appendix

\section{Unsimplified Graph Examples}\label{sec:appendixunsimpex}

\begin{figure}[ht]
\captionsetup[subfigure]{justification=Centering}
\begin{subfigure}[t]{0.22\textwidth}
    \includegraphics[width=\textwidth]{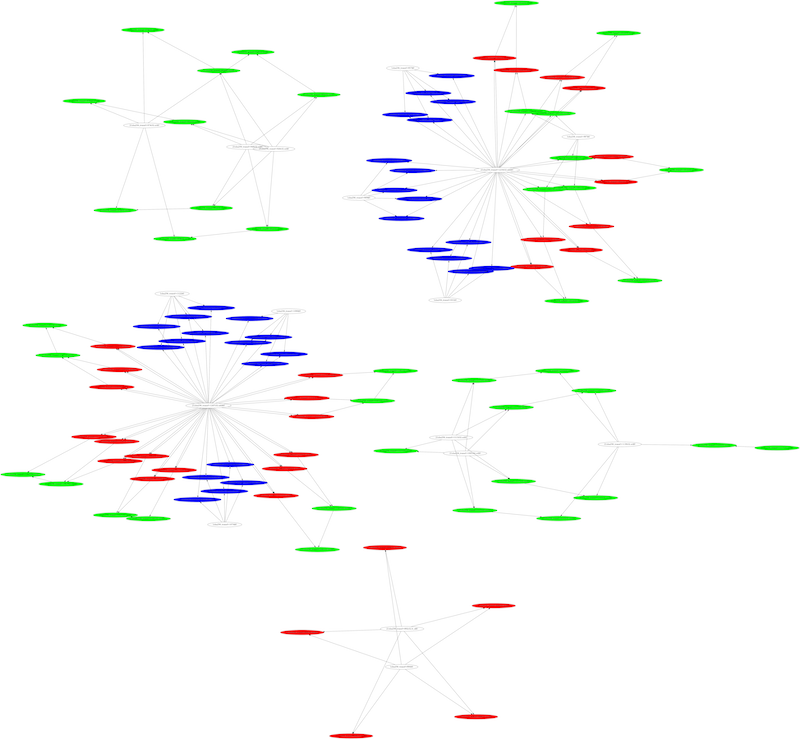}
    \caption{btc}
    \label{figure:btc0}
\end{subfigure}\hspace{\fill}
\begin{subfigure}[t]{0.22\textwidth}
    \includegraphics[width=\linewidth]{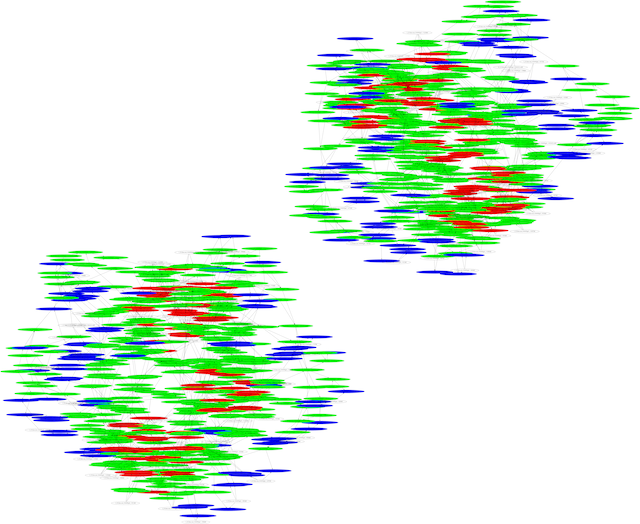}
    \caption{eth}
    \label{figure:eth0}
\end{subfigure}\hspace{\fill}
\begin{subfigure}[t]{0.22\textwidth}
    \includegraphics[width=\linewidth]{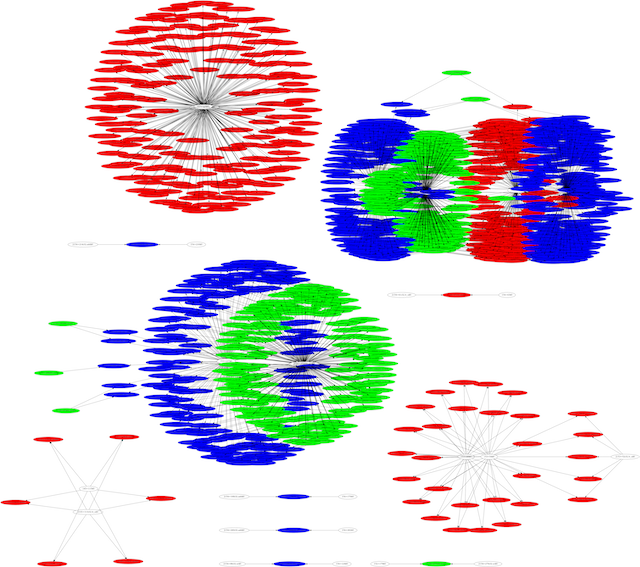}
    \caption{zny}
    \label{figure:zny0}
\end{subfigure}
\bigskip
\\
\begin{subfigure}[t]{0.22\textwidth}
    \includegraphics[width=\linewidth]{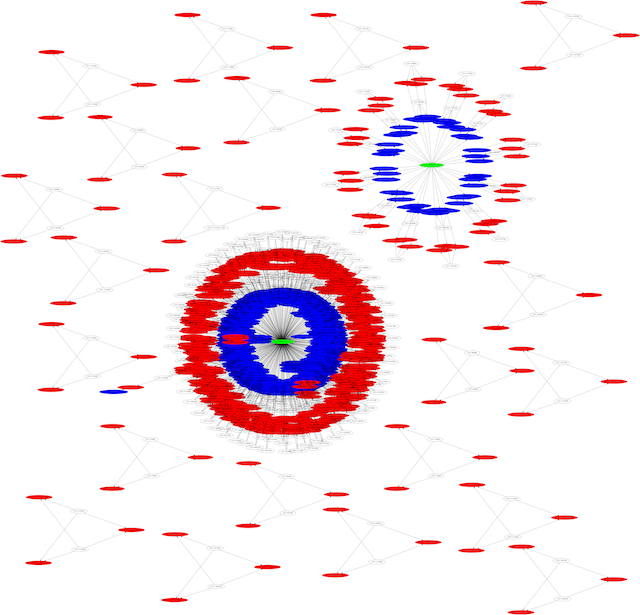}
    \caption{cn}
    \label{figure:cn0}
\end{subfigure}\hspace{\fill}
\begin{subfigure}[t]{0.22\textwidth}
    \includegraphics[width=\linewidth]{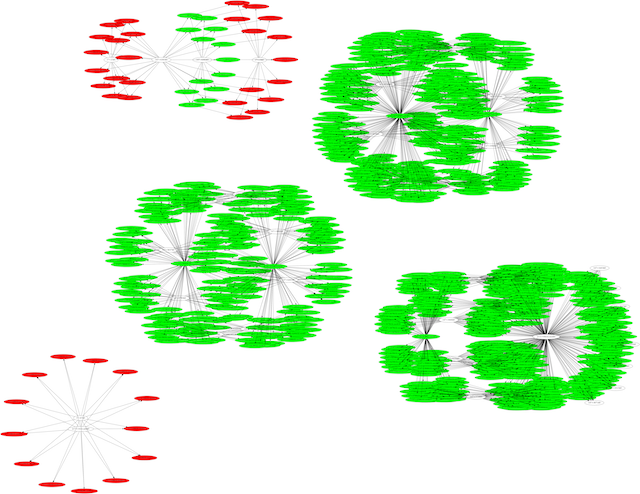}
    \caption{wmp}
    \label{figure:wmp0}
\end{subfigure}\hspace{\fill}
\begin{subfigure}[t]{0.22\textwidth}
    \includegraphics[width=\linewidth]{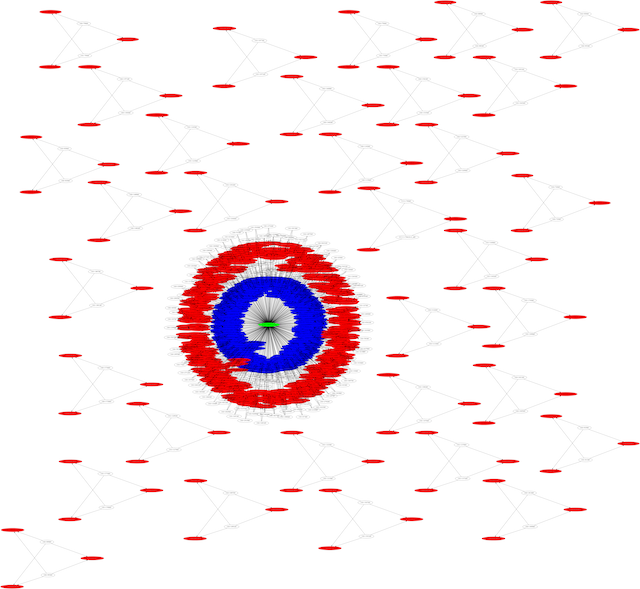}
    \caption{xmr}
    \label{figure:xmr0}
\end{subfigure}
\bigskip
\\
\begin{subfigure}[t]{0.22\textwidth}
    \includegraphics[width=\linewidth]{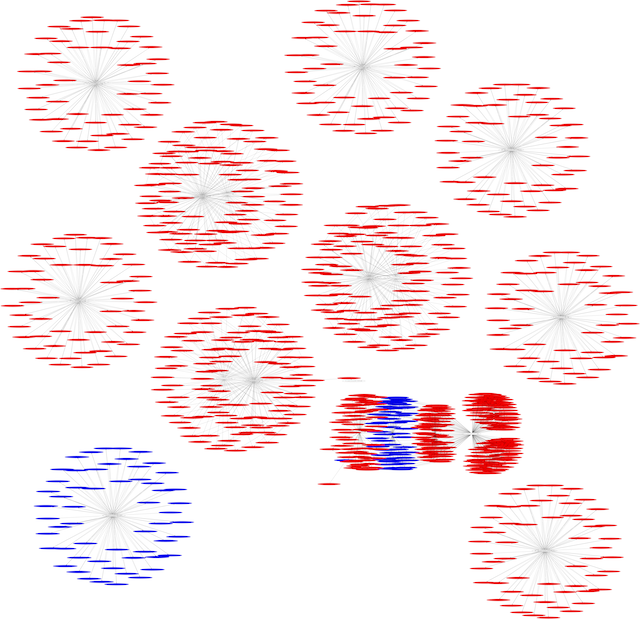}
    \caption{boa}
    \label{figure:boa0}
\end{subfigure}\hspace{\fill}
\begin{subfigure}[t]{0.22\textwidth}
    \includegraphics[width=\linewidth]{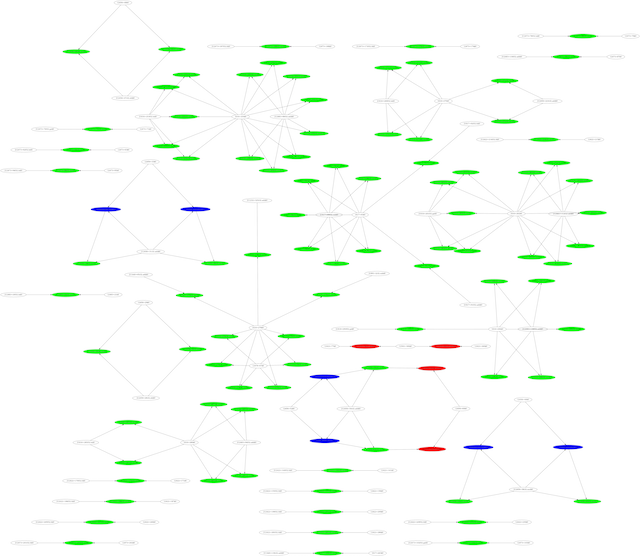}
    \caption{bullet}
    \label{figure:bullet0}
\end{subfigure}\hspace{\fill}
\begin{subfigure}[t]{0.22\textwidth}
    \includegraphics[width=\linewidth]{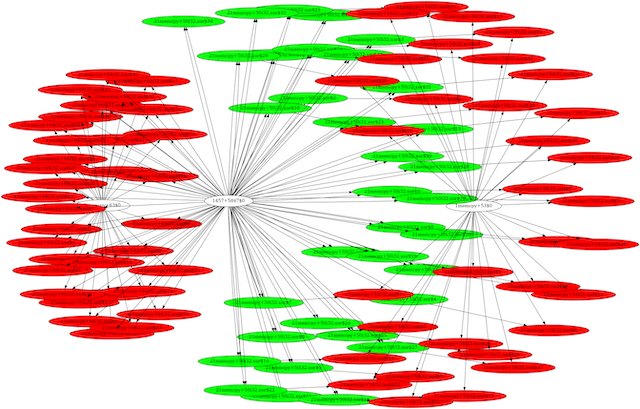}
    \caption{chocolatekeen}
    \label{figure:chocolatekeen0}
\end{subfigure}
\bigskip
\\
\begin{subfigure}[t]{0.22\textwidth}
    \includegraphics[width=\linewidth]{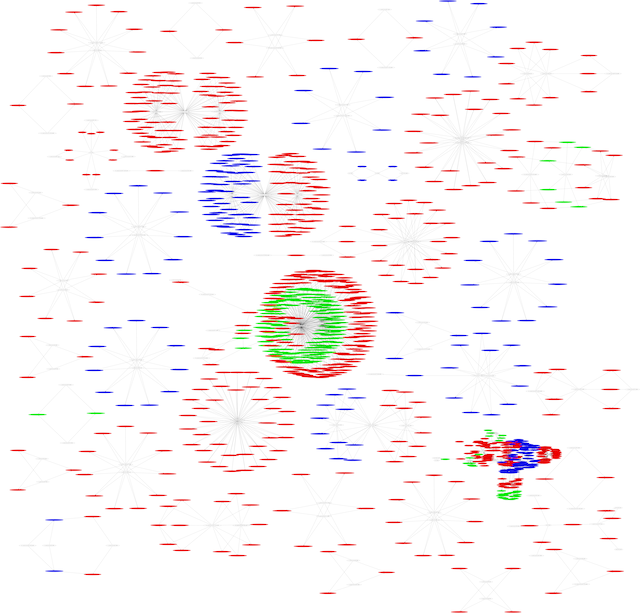}
    \caption{ffmpeg}
    \label{figure:ffmpeg0}
\end{subfigure}\hspace{\fill}
\begin{subfigure}[t]{0.22\textwidth}
    \includegraphics[width=\linewidth]{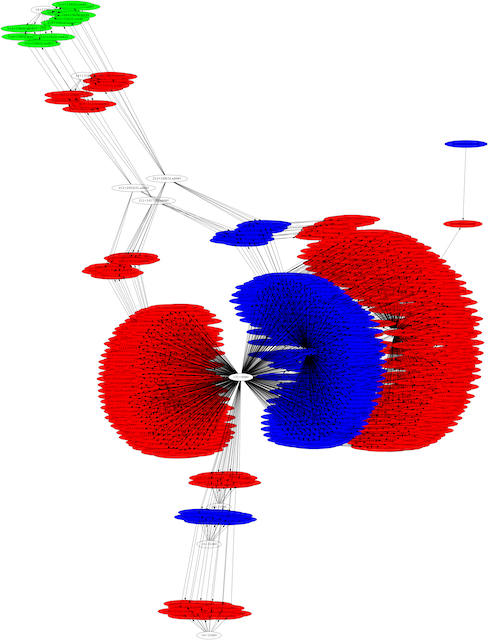}
    \caption{sandspiel}
    \label{figure:sandspiel0}
\end{subfigure}\hspace{\fill}
\begin{subfigure}[t]{0.22\textwidth}
    \includegraphics[width=\linewidth]{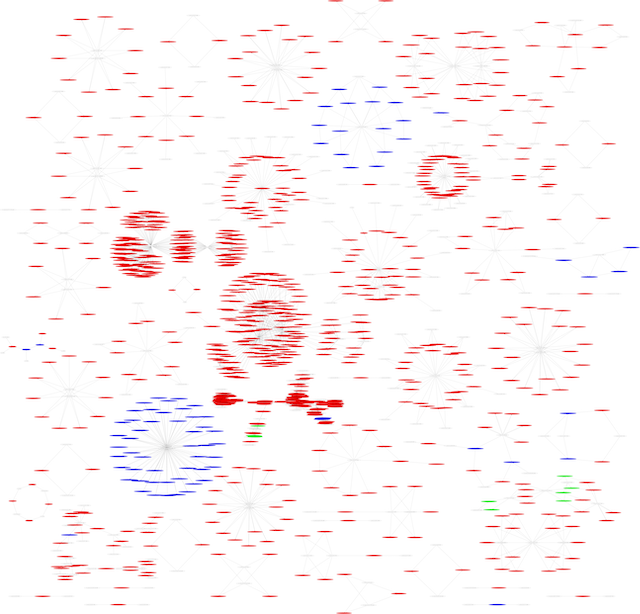}
    \caption{sqlgui}
    \label{figure:sqlgui0}
\end{subfigure}
\caption{\textbf{(Example Original Graphs of Miners and Non-miners)}}
\label{figure:originalgraphsexample}
\end{figure}

\end{document}